%% file: main.tex
\documentclass[11pt]{article}
\usepackage{amsmath,amsthm,amssymb}
\usepackage{mathtools}
\usepackage{soul}
\usepackage{graphicx}
\usepackage[hidelinks]{hyperref}
\usepackage{times}
\usepackage{mathtools}
\usepackage{thm-restate}
\usepackage{bm,bbm}
\usepackage{booktabs}
\usepackage{algorithm}
\usepackage{algorithmic}
\usepackage{xcolor}
\usepackage{subcaption}
\usepackage{todonotes}
\usepackage{centernot}

\newtheorem{definition}{Definition}
\newtheorem{theorem}{Theorem}
\newtheorem*{theorem*}{Theorem}
\newtheorem{lemma}{Lemma}

\newtheorem{example}{Example}
\newtheorem{assumption}{Assumption}

\newtheorem{claim}{Claim}

\newcommand{\alg}{\textnormal{Alg}}

\newcommand{\notshow}[1]{{}}

\newcommand{\pos}{\text{pos}}
\newcommand{\satcc}{\text{\sc -CC-Winner}}

\newcommand{\satm}{\text{\sc -Monroe-Winner}}

\newcommand{\cc}{\text{Chamberlin-Courant}}

\newcommand{\polynomial}{\text{\rm poly} }

\newcommand{\winner}[1]{\text{\sc #1-Winner}}

\newcommand{\wmg}{\text{WMG}}
\newcommand{\ks}{{\sc KemenyScore}}

\newcommand{\ds}{{\sc DodgsonScore}}
\newcommand{\ys}{{\sc YoungScore}}

\newcommand{\efas}{\text{\sc EFAS}}

\newcommand{\xc}{\text{\sc X3C}}
\newcommand{\topgroup}{\mathrm{Top}}
\newcommand{\ma}{\mathcal A}
\newcommand{\mm}{\mathcal M}
\newcommand{\mn}{\mathcal N}

\newcommand{\ml}{\mathcal L}

\newcommand{\kt}{\text{KT}}

\newcommand{\applast}{\mathbf{AppLast}}

\newcommand{\gcyc}{{G_{3c}}}

\newcommand{\citep}[1]{\cite{#1}}
\newcommand{\citet}[1]{\cite{#1}}
\begin{document}
\title{Beyond the Worst Case: Semi-Random Complexity Analysis of Winner Determination}
%
%
\author{Lirong Xia \\ Rensselaer Polytechnic Institute\\ {\tt xialirong@gmail.com}\and
Weiqiang Zheng\\Yale University\\ {\tt weiqiang.zheng@yale.edu}}

\date{}
%
\maketitle              
\begin{abstract}
The computational complexity of winner determination is a classical and important problem in computational social choice. Previous work based on worst-case analysis has established NP-hardness of winner determination for some classic voting rules, such as Kemeny, Dodgson, and Young. 

In this paper, we revisit the classical problem of winner determination through the lens of {\em semi-random analysis}, which is a worst average-case analysis where the preferences are generated from a distribution chosen by the adversary. Under a natural class of semi-random models that are inspired by recommender systems, we prove that winner determination remains hard for Dodgson, Young, and some multi-winner rules such as the \cc{} rule and the Monroe rule. Under another natural class of semi-random models that are extensions of the Impartial Culture, we show that winner determination is hard for Kemeny, but is easy for Dodgson. This illustrates an interesting separation between Kemeny and Dodgson.
\end{abstract}

\input{intro}
\input{model}
\input{ds-ys-hard}

\input{ks-ds}

\section{Conclusion}
In this paper, we conduct semi-random complexity analysis of winner determination under various voting rules. We give the first semi-random complexity results for the Dodgson rule, the Young rule, the Chamberlin-Courant rule, and the Monroe rule. We also prove a hardness result for the Kemeny rule and a semi-random easiness result for the Dodgson rule, illustrating an interesting separation between the semi-random complexity of winner determination under different NP-hard voting rules.

As for future direction, an ambitious goal is to develop a dichotomy theorem for the semi-random complexity of winner determination: winner determination is efficient if and only if the semi-random model satisfies certain conditions. The semi-random complexity of winner determination under models beyond Assumption \ref{asmpt:general-hard} is a natural and interesting problem. We also conjecture that under the average-case analysis, \ys{} is easy to decide with high probability but \ks{} remains hard.

\section*{Acknowledgements}
We thank anonymous reviewers for helpful feedback and suggestions. LX acknowledges NSF \#1453542  and a gift fund from Google for support.

%
%
%
\bibliographystyle{plain}
\bibliography{ref4,new}

\clearpage
\appendix
\input{appendix}

\end{document}

%% file: intro.tex
\section{Introduction}
Voting is one of the most popular methods for group decision-making. In large-scale, high-frequency group decision-making scenarios, it is highly desirable that the winner can be computed in a short amount of time. 
The complexity of \emph{winner determination} under common voting rules is thus not only a classic theoretical problem in  computational social choice~\cite[chapter 4, 5]{Brandt2016:Handbook}, but also an important consideration in practice.

In this paper, we focus on several classic voting rules: the Kemeny rule, the Dodgson rule, and the Young rule, whose winner determination problems are denoted as \ks{}, \ds{}, and \ys{}, respectively. The Kemeny rule, which is closely related to the Feedback Arc Set problem~\citep{Alon06:Ranking,Ailon2008:Aggregating}, is a classical method for recommender systems  and information retrieval~\citep{Dwork01:Rank}. The Dodgson rule and the Young rule have also been extensively studied in the literature~\citep{Hemaspaandra97:Exact,Rothe2003:Exact,Caragiannis09:Approximability,Brandt2016:Handbook}.

Previous work has established the (worst-case) NP-hardness of winner determination under the Kemeny rule, the Dodgson rule, and the Young rule~\citep{Bartholdi89:Voting,Rothe2003:Exact}. 
Using average-case analysis, McCabe-Dansted et al. \cite{McCabe-Dansted2008:Approximability} and Homan and Hemaspaandra \cite{Homan2009:Guarantees} showed that \ds{} admits an efficient algorithm that succeeds with high probability, where each ranking is generated i.i.d.~uniformly,  known as the \emph{Impartial Culture (IC)} assumption in social choice. Unfortunately, IC or generally any i.i.d. distribution has been widely criticized of being unrealistic (see, e.g.,  \cite[p.~30]{Nurmi1999:Voting}, \cite[p.~104]{Gehrlein2006:Condorcets}, and~\citep{Lehtinen2007:Unrealistic}). It remains unknown whether there exists an efficient algorithm for \ds{} beyond IC. This motivates us to ask the following question: 
\begin{equation*}
    \text{\emph{What is the complexity of winner determination beyond worst-case analysis and IC?}}
\end{equation*}

One promising idea is to 
tackle this question through the lens of \emph{smoothed complexity analysis} \citep{Spielman2009:Smoothed,Baumeister2020:Towards}, a beautiful and powerful framework for analyzing the performance of algorithms in practice. Smoothed analysis can be seen as a worst average-case analysis, where the adversary first arbitrarily chooses an instance, and then Nature adds random noise (perturbation) to it, based on which the expected runtime of an algorithm is evaluated. Smoothed analysis explains why the simplex method is fast despite its worst-case exponential time complexity~\citep{Spielman2004:Smoothed}. It has been successfully applied to many fields to understand the practical performance of algorithms, see the survey by Spielman and Teng~\cite{Spielman2009:Smoothed}.


Smoothed analysis belongs to the more general approach of complexity analysis under {\em semi-random models}~\citep{Blum1990:Some,Feige2021:Introduction}, where the problem instance contains an adversarial component and a random component. In this paper, we adopt the semi-random  model called the \emph{single-agent preference model} \citep{Xia2020:The-Smoothed}, where the adversary chooses a preference distribution for each agent from a set $\Pi$ of distributions. Note that if $\Pi$ consists of only the uniform distribution, then the model is equivalent to IC. By varying $\Pi$, the model can provide a smooth transition from average-case analysis to worst-case analysis.  Under this model, Xia and Zheng~\cite{Xia2021:The-Smoothed} proved the semi-random hardness of computing Kemeny ranking and Slater ranking with mild assumptions. 
However, their hardness results do  not imply hardness of \ks{} under the same model, because \ks{} is easier than computing the Kemeny ranking (see Definition \ref{dfn:ks}).
The semi-random complexity of the Dodgson rule and the Young rule were also left as open questions \citep{Xia2021:The-Smoothed}. 

\paragraph{\bf Our contributions.}
We provide the first set of results on the computational complexity of winner determination  under the following two classes of semi-random models.

The first class of models are inspired by recommender systems and information retrieval, where the number of alternatives $m$ can be very large and it is inefficient for an intelligent system to learn the total ranking. In such cases, one often uses top-$K$ ranking algorithms that only recover the top-$K$ ranking with high accuracy for $K = o(m)$~\citep{pmlr-v70-mohajer17a,chen2018nearly}. Similarly in social choice, the collected preference from an agent is more robust over her few top-ranked alternatives and may be much more noisy over the remaining alternatives (see Example \ref{ex:partial alternative model}). Formally, we capture such features in Assumption \ref{asmpt:general-hard}.
Then, we prove in Theorem~\ref{thm:hard-ds} and Theorem \ref{thm:hard-ys} that \ds{} and \ys{} remain hard under Assumption \ref{asmpt:general-hard} unless NP=ZPP.  
Similar semi-random hardness results  also hold for some \emph{multi-winner} rules, i.e., the \cc{} rule and the Monroe rule (Theorem~\ref{thm:hard-cc-monroe}).
    
The second class of semi-random models   are called $\alpha$-Impartial Culture ($\alpha$-IC for short, see Definition~\ref{dfn:alpha-IC}) where $\alpha \in [0,1]$. They are a relaxation of IC such that a single ranking receives probability $1-\alpha$ and the other rankings are uniformly distributed. When $\alpha$ is $\frac{1}{\mathcal{O}(\text{poly}(m))}$  away from $1$, we 
illustrate an interesting separation between the complexity of \ks{} and that of \ds{}: winner determination is hard for \ks{} (Theorem~\ref{thm:hard-ks}) while being easy for \ds{} (Theorem~\ref{thm:dodgson-easy-alphaIC}).
    

\subsection{Related Works and Discussions}

\paragraph{\bf Smoothed and semi-random analysis.} 
Semi-random models have been widely adopted to analyse the performance of algorithms in practice and to circumvent worst-case computational hardness in the field of combinatorial optimization \citep{blum1995coloring,feige1998heuristics}, mathematical programming \citep{Spielman2004:Smoothed}, and recently in algorithmic game theory \citep{psomas2019smoothed,boodaghians2020smoothed,smoothed_focs_2020,blum2021incentive,bai2022fair}. 
We refer the readers to recent surveys of semi-random models \citep{Feige2021:Introduction} and beyond worst-case analysis \citep{roughgarden2021beyond} for a comprehensive literature review. We mention here that the partial alternative randomization model in Example \ref{ex:partial alternative model} is inspired by the partial bit randomization model which has been applied to smoothed complexity analysis~\citep{goos_smoothed_2003} and smoothed competitive ratio analysis~\citep{becchetti_average-case_2006}.

Recently, semi-random analysis has also been proposed in the field of social choice \citep{Baumeister2020:Towards,Xia2020:The-Smoothed}.The smoothed probability of paradoxes and ties, and strategyproofness in voting are studied \citep{Xia2020:The-Smoothed,xia2021likely,xia2021semi,ding2021approximately}. As mentioned above, Xia and Zheng~\cite{Xia2021:The-Smoothed} studied complexity of computing Kemeny and Slater rankings under semi-random models. We are not aware of other semi-random complexity results in computational social choice, which motivates this work. 

Beier and V{\"o}cking~\cite{beier_typical_2006} studied the case of the integer linear programs (ILPs) over the unit cube and showed that a problem has polynomial smoothed complexity if and only if it admits a pseudo-polynomial algorithm. Since winner determination under voting rules studied in this paper can also be formulated as ILPs, one might be tempted to think that  the results in \citep{beier_typical_2006} also apply to the single-agent preference model.  However, this is not true because
they only considered continuous perturbation for \emph{real numbers}, while the set of rankings is \emph{discrete}. Their conclusion works for discrete combinatorial optimization problems only if the continuous noise is added to the so-called \emph{stochastic parameters} that are real numbers, so that the problem's combinatorial structure remains unchanged, which is not the case of our setting.  

\notshow{
Smoothed complexity analysis was recently introduced to computational social choice by ~\cite{Baumeister2020:Towards} and
\cite{Xia2021:The-Smoothed} studied smoothed/semi-random complexity of computing Kemeny and Slater rankings. We are not aware of other semi-random complexity results in computational social choice, which motivates this work. 

\cite{beier_typical_2006} studied the case of the integer linear programs (ILPs) over the unit cube and showed that a problem has polynomial smoothed complexity if and only if it admits a pseudo-polynomial algorithm. Since winner determination under voting rules studied in this paper can also be formulated as ILPs, one might be tempted to think that  the results in \citep{beier_typical_2006} also apply to the single-agent preference model.  However, this is not true because
\cite{beier_typical_2006} only considered continuous perturbation for \emph{real numbers}, while the set of rankings is \emph{discrete}. Their conclusion works for discrete combinatorial optimization problems only if the continuous noise is added to the so-called \emph{stochastic parameters} that are real numbers, so that the problem's combinatorial structure remains unchanged, which is not the case of our setting.  

There is a large body of literature on semi-random analysis and smoothed analysis in mathematical programming, machine learning, numerical analysis, discrete math, combinatorial optimization, etc. 
We refer the readers to~\citep{roughgarden2021beyond} for a recent survey. The partial alternative randomization model in Example \ref{ex:partial alternative model} is inspired by the partial bit randomization model which has been applied to smoothed complexity analysis~\citep{goos_smoothed_2003} and smoothed competitive ratio analysis~\citep{becchetti_average-case_2006}.}


\paragraph{\bf Complexity of winner determination.} 
There is a large body of literature on worst-case computational complexity of winner determination under various voting rules. Bartholdi et al.~\cite{Bartholdi89:Voting} proved that computing \ds{} and \ys{} are NP-hard, respectively. They also provided the NP-completeness of \ks{}, which holds even for only four voters~\citep{Dwork01:Rank}. The problem of computing Dodgson winner, Young ranking, and Kemeny ranking were proved to be   $\Theta_2^{\text{P}}$ complete \citep{Hemaspaandra97:Exact,Rothe2003:Exact,Hemaspaandra05:Complexity}. 



%% file: model.tex
\section{Model and Preliminaries}\label{sec:model}
\paragraph{\bf Basics of voting.}
Let $\ma_m=\{a_1,\ldots,a_m\}$ denote the set of $m$ alternatives and $\ml(\ma_m)$ the set of rankings (linear orders) over $\ma_m$.
A {\em (preference) profile} $P\in \ml(\ma_m)^n$ is a collection of $n$ agents' rankings, which is also called their {\em preferences}. Throughout the paper, we assume without loss of generality that $m \ge 3$ since winner determination is easy for 2 alternatives. For any ranking $R \in \ml(\ma_m)$, we denote  $\topgroup_K(R)$ the top-$K$ ranking of $R$.  For a permutation $\sigma$ over $\ma_m$ and any distribution $\pi$ over $\ml(\ma_m)$, we denote $\sigma(\pi)$ the permuted distribution where $\Pr_{\sigma(\pi)}(\sigma(R)) = \Pr_\pi(R)$ for all $R \in \ml(\ma_m)$.

\paragraph{\bf The Dodgson rule, the Young rule and the Kemeny rule.}
The Condorcet winner of preference profile $P$ is defined as the alternative $a \in \ma_m$ who is preferred to every $b \in \ma_m$ by strictly more than half of the agents. 
The {\em Dodgson score} of $a$ in $P$ is defined as the smallest number of sequential exchanges of adjacent alternatives in rankings of $P$ to make  $a$ the Condorcet winner. 
The {\em Young score} of $a$ in $P$ is defined as the size of the largest subset of preferences where $a$ is the Condorcet winner. The Dodgson rule chooses the alternatives with the lowest Dodgson score as winners, and the Young rule chooses the alternatives with the highest Young score as  winners. The winner determination problems of the Dodgson rule and the Young rule are defined as follows.

\begin{definition}[ \ds{} and \ys{}] Given $P\in \ml(\ma_m)^n$, $a \in \ma_m$, and $t \in \mathbb{N}$, in \ds{} (respectively, \ys{}), we are asked to decide whether the Dodgson score (respectively, Young score) of $a$ in $P$ is at most (respectively, at least) $t$. 
\end{definition}

The {\em Kendall's Tau distance (KT distance)} between two linear orders $R,R'\in\ml(\ma)$, denoted by $\kt(R,R')$,  is the number of pairwise disagreements between $R$ and $R'$. Given a  profile $P$ and a linear order $R$, the {\em KT distance} between $R$ and $P$ is defined to be $ \kt(P,R) = \sum_{R'\in P}\kt(R,R')$. The {\em Kemeny score} of an alternative $a$ in $P$ is defined as the minimum KT distance between any linear order that ranks $a$ at the top. The Kemeny rule chooses the alternatives with the lowest Kemeny score. Besides, the Kemeny ranking is defined as the ranking with minimum KT distance to $P$. The winner determination problem of the Kemeny rule is defined as follows:
\begin{definition}[\bf \ks{}]\label{dfn:ks}
Given $P\in \ml(\ma_m)^n$ and $t\in \mathbb{N}$, in \ks{}, we are asked to decide if there exists an alternative $a \in \ma_m$ whose Kemeny score is at most $t$. 
\end{definition}
If we can compute the Kemeny ranking, then we can compute its KT distance to $P$ in polynomial time and then decides \ks{}. Thus \ks{} is easier than computing the Kemeny ranking.

\paragraph{\bf Semi-random complexity analysis.}
We use the following semi-random model, proposed in \citep{Xia2020:The-Smoothed} and used for semi-radom complexity analysis in \citep{Xia2021:The-Smoothed}.

\begin{definition}[\bf Single-agent preference model~\citep{Xia2020:The-Smoothed}]A {\em single-agent preference model} for $m$ alternatives is denoted by $\mm_m=(\Theta_m,\ml(\ma_m),\Pi_m)$.  $\Pi_m$ is a set of distributions over $\ml(\ma_m)$ indexed by a parameter space $\Theta_m$ such that for each parameter $\theta \in \Theta_m$, $\pi_{\theta} \in \Pi_m$ is its corresponding distribution. 
\end{definition}
We say $\mm_m$ is {\em P-samplable} if there exists a poly-time sampling algorithm for each distribution in $\Pi_m$. It is the ``most natural restriction'' on general distributions, which is less restrictive than the commonly-studied {\em P-computable} distributions~\cite[p.~17,18]{Bogdanov2006:Average-Case}.
We say $\mm_m$ is {\em neutral} if for any $\pi\in \Pi_m$ and any permutation $\sigma$ over $\ma_m$, we have $\sigma(\pi)\in \Pi_m$.  
Note that winner determination under all the above voting rules is in P when $m$ is bounded above by a constant. Therefore, we are given a sequence of single-agent preference models $\vec\mm=\{\mm_m=(\Theta_m,\ml(\ma_m),\Pi_m):m\ge 3\}$. We say $\vec\mm$ is P-samplable (respectively, neutral) if $\Pi_m$ is P-samplable (respectively, neutral) for any $m\ge3$. 

We introduce the following generalization of the Impartial Culture model, which is P-samplable and neutral. 
\begin{definition}[\bf $\alpha$-Impartial Culture]\label{dfn:alpha-IC}
Fix $\alpha \in [0,1]$. {\em $\alpha$-Impartial Culture ($\alpha$-IC)} is a single-agent preference model $\mm_m=(\Theta_m,\ml(\ma_m),\Pi_m)$ such that $\Theta_m = \ml(\ma_m)$ and for each $R \in \ml(\ma_m)$, distribution $\pi_R$ is defined as \[\Pr_{R'\sim \pi_R}[R'] = \frac{\alpha}{m!} + (1-\alpha)\bm 1[R'=R],\]
where $\bm 1[R'=R] = 1$ if $R = R'$ and $\bm 1[R'=R] = 0$ otherwise.
Fix $\vec\alpha = (\alpha_m)_{m\ge 3}$ such that $\alpha_m \in [0,1]$ for all $m \ge 3$. Denote $\vec\alpha$-IC the sequence of models $\{\alpha_m\text{-IC}:m\ge3\}$. It is easy to see that $\vec\alpha$-IC is P-samplable and neutral.
\end{definition}

The semi-random profile $P$ according to $\mm_m$ is generated as follows. First, the adversary chooses $\vec\pi=(\pi_1,\ldots,\pi_n)\in \Pi_m^n$. Then agent $j$'s ranking will be independently (but not necessarily identically) generated from $\pi_j$ for any $j \in [n]$. The semi-random version of winner determination under $\vec\mm$ is defined as follows, which is similar to the definition in a recent paper on smoothed hardness of two-player Nash equilibrium~\citep{smoothed_focs_2020}.

\begin{definition}[\bf {\sc Semi-Random}-\ds{}]\label{dfn:smoothed-winner}
Fix a sequence of single-agent preference models $\vec\mm$. Given alternative $a \in \ma_m$, $t \in \mathbb{N}$ and a semi-random profile $P$ drawn from $\mm_m$, we are asked to decide whether the Dodgson score of $a$ is at most $t$, with probability at least $1-\frac{1}{m}$.\footnote{The algorithm is allowed to return ``Failure" with probability at most $\frac 1m$. However, when it returns YES or NO, the answer must be correct. 
Our hardness results hold even for algorithms that are only required to succeed with probability $o(1)$.
}
\end{definition}

\begin{definition}[\bf {\sc Semi-Random}-\ks{}]
Fix a sequence of single-agent preference models $\vec\mm$. Given $t \in \mathbb{N}$ and a semi-random profile $P$ drawn from $\mm_m$, we are asked to decide whether there exists an alternative whose Kemeny score of $a$ is at most $t$, with probability at least $1-\frac{1}{m}$.
\end{definition}
The definition of {\sc Semi-Random}-\ys{} is similar (See Appendix  \ref{appendix:young}). 
  

%% file: ds-ys-hard.tex
\section{Semi-Random Hardness of \ds{} and \ys{}}\label{sec:semi-hardnes} 

In many applications, such as recommender systems and information retrieval, the number of alternatives $m$ can be very large and it is inefficient for an intelligent system to learn the total ranking. In such cases, one often uses Top-$K$ ranking algorithms which only recover the top-$K$ ranking with high accuracy for $K = o(m)$~\citep{pmlr-v70-mohajer17a,chen2018nearly}. Similarly, the collected preference from an agent is more robust over her few top-ranked alternatives and can be much more noisy over the remaining alternatives. Such features are captured by 
Assumption \ref{asmpt:general-hard} below. Informally, Assumption \ref{asmpt:general-hard} states that there \emph{exists} a distribution in $\Pi_m$ that does not significantly ``perturb'' one top-$K$ ranking for $K = \Theta(m^{\frac{1}{d}})$ where $d \ge 1$. 
\begin{assumption}[\bf Top-$K$ concentration]\label{asmpt:general-hard} A series of single-agent preference models $\vec\mm$ is P-samplable, neutral, and satisfies the following condition: there exists a constant $d > 1$ such that for any sufficiently large $m$ and $K = \lceil m^{\frac{1}{d}}\rceil$, there exists $\ma'\subseteq \ma_m$, $R'\in \ml(\ma')$, and $\pi\in \Pi_m$, such that $|\ma'|=K$ and 
$$\Pr_{R\sim \pi}(\text{Top}_K(R) = R') \ge 1-\frac{1}{K}.$$ 
\end{assumption}
The following \emph{partial alternative randomization model}, in the spirit of partial bit randomization model \citep{goos_smoothed_2003,becchetti_average-case_2006}, satisfies Assumption \ref{asmpt:general-hard}. The partial bit randomization model applies to $m$-bits non-negative integer by randomly flipping its $m-K$ least significant bits while keeping its $K$ most significant bits unchanged.
\begin{example}
\label{ex:partial alternative model} The {\em partial alternative randomization model} is denoted by $\mm_m(K)$ and has parameter space $\ml(\ma_m)$. For any $R \in \ml(\ma_m)$, the distribution $\pi_R$ is obtained by uniformly at random perturbing the order of the $m-K$ least preferred alternatives in $R$ and keeping the top-$K$ ranking unchanged. For any constant $d$ and $K \ge m^{\frac{1}{d}}$, the model is P-samplable, neutral and satisfies Assumption~\ref{asmpt:general-hard}. Note that in such model, each ranking receives probability at most $\frac{1}{(m-K)!} = \frac{1}{\Omega(\exp{m})}$.
\end{example}

We show that for models that satisfying Assumption \ref{asmpt:general-hard}, winner determination under the Dodgson rule and the Young rule is hard unless NP=ZPP. Note that NP$\ne$ZPP is widely believed to hold in complexity theory. The high-level idea is to combine the existence of a top-$K$ concentration distribution guaranteed by Assumption \ref{asmpt:general-hard} and neutrality, to show that for any possible input of a NP-complete problem, the adversary is able to construct a distribution of voting profile such that efficient semi-random winner determination implies a coRP algorithm for the NP-complete problem. Thus NP$\subseteq$coRP and it implies NP$=$ZPP by the following reasoning.  Recall that RP $\subseteq$ NP. Therefore, RP $\subseteq${\sc NP}$\subseteq$ coRP, which means that RP $=$ RP $\cap$ coRP. Recall that RP $\cap$ coRP $=${\sc ZPP}. We have RP $=$ ZPP, which means that coRP $=$ coZPP. Since  coZPP $=$ ZPP, it follows that {\sc NP}$\subseteq$ coRP $=$ coZPP $=${\sc ZPP}.

\begin{theorem}[\bf Semi-random hardness of \ds{}]\label{thm:hard-ds} For any serie of single-agent preference models $\vec\mm$ that satisfies Assumption~\ref{asmpt:general-hard}, there exists no polynomial-time algorithm for {\sc Semi-Random}-\ds{} under $\vec\mm$ unless {\sc NP}$=${\sc ZPP}.
\end{theorem}

\begin{proof}
\notshow{
The proof is based on the following observation of the Dodgson rule. We introduce one more notation here. For any profile $P \in \ml(\ma_m)^n$, we denote $\applast(P,m')$ the set of profiles obtained from $P$ by appending $m'$ extra alternatives to the bottom of each agent's preferences in any order. It follows that $|\applast(P,m')| = (m'!)^n$. 
\begin{lemma}\label{lem:dodgson-toprobust}
For any profile $P_1 \in \ml(\ma_m)^n$, any integer $m'\ge 1$ and profile $P_2 \in \applast(P_1,m')$, the following holds for any alternative $a \in \ma_{m}$:
\begin{itemize}
    \item If $a$ is Dodgson winner in $P_1$, then $a$ is also Dodgson winner in $P_2$.
    \item The Dodgson score of $a$ in $P_1$ is equal to that in $P_2$.
\end{itemize}
\end{lemma}
The proof of Lemma \ref{lem:dodgson-toprobust} follows by definition and can be found in Appendix \ref{appendix:ds-toprobust}. Informally, Lemma \ref{lem:dodgson-toprobust} states that by appending alternatives at the bottom of each agent's preference order, the winner and score under the Dodgson rule are robust. 
}
\noindent{\bf Overview of the proof.} We leverage the reduction in \citep{Bartholdi89:Voting} that reduces the NP-complete problem {\sc Exact Cover by 3-Sets} (\xc{}) to \ds{}. An instance of $\xc$ is denoted by $(U,S)$ including a $q$-element set $U$ such that $q$ is divisible by 3 and a collection $S$ of 3-element subsets of $U$. We are asked to decide whether $S$ contains an exact cover for $U$, i.e., a subcollection $S'$ of $S$ such that every element of $U$ occurs in exactly one member of $S'$. 

Suppose that {\sc Semi-Random}-\ds{} has a polynomial-time algorithm, denoted as \alg{}. We will use \alg{} to construct a coRP algorithm for \xc{}. 
Formally, the proof proceeds in two steps. For any instance of \xc{}, in Step 1, we follow the original reduction to construct a profile $P_1$. Then we construct a parameter profile $P^{\Theta}$ using the semi-random model $\vec\mm$ based on $P_1$. Note that a parameter profile corresponds to distribution over profiles. In Step 2, we show that \alg{} can be leveraged to Algorithm \ref{alg:rpalg} to prove that \xc{} is in coRP, which implies {\sc NP}$=$ ZPP as shown above.

Let $(U,S)$ be any instance of \xc{} such that $U = \{u_1,u_2,\cdots,u_q\}$ and $S = \{S_1,S_2,\cdots,S_s\}$ is a collection of $s$ distinct 3-element subsets of $U$. We assume without loss of generality that $q/3\le s\le q^3/6$ because $(U,S)$ must be a NO instance if $s < q/3$ and there are at most $\binom{q}{3} \le q^3/6$ distinct 3-element subsets of $U$. 
\paragraph{\bf Step 1. Construct profile $P_1$ and parameter profile $P^\Theta$.}  We first use the reduction by Bartholdi et al.~\cite{Bartholdi89:Voting} to construct a voting profile $P_1 \in \ml(\ma_{m_1})^n$ of polynomial-size in $q$. The proof of Lemma~\ref{lem:dodgson-reduction} can be found in Appendix~\ref{appendix:ds-reduction}.
\begin{lemma}\label{lem:dodgson-reduction}
We can construct a profile $P_1 \in \ml(\ma_{m_1})^n$ with $m_1 = 2q+s+1 = \mathcal{O}(q^3)$, $n \le 2(q+1)s+1 = \mathcal{O}(q^4)$, and an alternative $c$ such that $(P_1,c,\frac{4q}{3})$ is a YES instance of \ds{} if and only if $(U,S)$ is a YES instance of \xc{}. The construction can be done in polynomial time in $q$.
\end{lemma}
\notshow{
\begin{proof}
We first present the construction of profile $P_1$ and then we show why it satisfies the desired properties.
\begin{description}
    \item[Construction of $P_1$.] We first construct the set of alternatives $\ma_{m_1}$, which contains three type of alternatives. $\ma_{m_1}$ contains a \emph{critical alternative} $c$. For any $i\in [q]$ and element $u_i \in U$, $\ma_{m_1}$ contains two \emph{element alternatives} $a_i$ and $b_i$. We denote by $E = \{a_i,b_i| i\in \{1,\cdots,q\}\}$ the set of element alternatives. For any $j \in [s]$ and subset $S_j \in S$, $\ma_{m_1}$ contains one \emph{subset alternative} $s_j$. We denote by $H$ the set of subset alternatives.
    
    Now we construct the ranking profile $P_1$, which consists of the following three sub-profiles.
    \begin{description}
        \item[1. Swing Rankings] We create $s$ rankings for each member of $S$. For each subset $S_j = \{u_{j_1},u_{j_2},u_{j_3}\}$, denote $E_j = \{a_{j_1},a_{j_2},a_{j_3}\}$. Let $R_{S_j}$ be any ranking of the form $E_j \succ s_j \succ c \succ (E/\ E_j)\cup (H/\ s_j)$ where the order of alternatives in each part can be arbitrary. 
        We set $P_{1,1}$ be the profile containing $s$ swing rankings $R_{S_j}$ for all $j \in [s]$. It is easy to see $|P_{1,1}| = s$.

        The idea behind swing rankings $P_{1,1}$ is the following. Note that switching the special alternative $c$ up 1 position in $R_{S_j}$ gains 0 vote for $c$ against all element alternatives $E$; switching 2 times gains 1 vote; switching 3 times get 2 votes; switching 4 times get 3 votes. Thus, among the swing rankings $P_{1,1}$, any additional votes for $c$ over element alternatives in $E$ require $4/3$ switches per vote on the average. Moreover, to achieve $4/3$ switches per vote, $c$ must be switched 4 times to the very top in each switched ranking. 

        \item[2. Equalizing Rankings.] For $i \in [q]$, let $N_i = |\{S_j \in S | u_i \in S_j\}|$ be the number of subsets in $S$ that contains $u_i$. Let $N^* = \max\{N_1,N_2,\cdots,N_q\}$. Let $R_{u_i}$ be any linear order of the form $a_i \succ b_i \succ c \succ (E/\ \{a_i,b_i\}) \cup H$, where the order of the alternatives after $c$ can be arbitrary. We set $P_{1,2}$ to be the profile containing $N^{*} - N_i$ copies of $R_{u_i}$ for all $i \in [q]$. We also have $|P_{1,2}| \le \sum_{i \in [q]} N^*-N_i \le q N^* \le qs$.

        By adding equalizing rankings in $P_{1,2}$ to $P_{1,1}$, each alternative $a_i$ gets equal score in the pairwise competitions against $c$. Note that among equalizing rankings $P_{1,2}$, additional votes for $c$ over an element alternative $a_i$ require at least 2 switches per vote on the average.

        \item[3. Incremental Rankings.] Let $R_I$ be any ranking of the form $a_1 \succ \cdots \succ a_q \succ b_1 \succ \cdots \succ b_q \succ c \succ H $ where the order of the alternatives after $c$ can be arbitrary. We set $P_{1,3}$ the profile containing $N_I$ copies of $R_I$ such that $a_i$ will defeat $c$ by exactly 1 voter in $P_{1,1} \cup P_{1,2} \cup P_{1,3}$ for any $i \in [q]$. It is easy to see that $|P_{1,3}| \le |P_{1,1}| + |P_{1,2}| + 1 \le s+qs+1$.

        By adding incremental rankings $P_{1,3}$ to $P_{1,1}\cup P_{1,2}$, alternative $a_i$ defeats $c$ by exactly 1 vote for all $i \in [q]$. Besides, additional votes for $c$ over an element alternative $a_i$ among incremental rankings $P_{1,3}$ require at least 2 switches per vote in average. 
    \end{description}
    We set $P_1 = P_{1,1} \cup P_{1,2} \cup P_{1,3}$. According to the construction, we know that $|P_1| \le \sum_{i=1}^3 |P_{1,i}| \le 2(q+1)s+1 = \mathcal{O}(q^4)$. 
    \item[Reduction.] Recall that each element alternative $a_i$ wins exactly 1 vote against the critical alternative $c$ in $P_1$. Thus in order to make $c$ the Condorcet winner, $c$ must win against each $a_i$. However, this requires at least $4q/3$ switches and is achievable only if i) all switches are among swing rankings, and ii) each switched swing ranking move $c$ to the top of the preference by 4 switches. It is obvious from the construction of swing rankings that any collection of swing rankings that can elect $c$ by no more than $4q/3$ switches correspond to an exact 3-cover of $(U,S)$. Thus $(P_1,c,\frac{4q}{3})$ is a YES instance of \ds{} if and only if $(U,S)$ is a YES instance of \xc{}.
\end{description}
It is clear that the construction can be done in polynomial time in $q$. This completes the proof.
\end{proof}
}
The following observation of the Dodgson rule is crucial for the proof. We introduce one more notation here. For any profile $P \in \ml(\ma_m)^n$, we denote $\applast(P,m')$ the set of profiles obtained from $P$ by appending $m'$ extra alternatives to the bottom of each agent's preferences in any order. It follows that $|\applast(P,m')| = (m'!)^n$. 
\begin{lemma}\label{lem:dodgson-toprobust}
For any profile $P_1 \in \ml(\ma_m)^n$, any integer $m'\ge 1$ and profile $P_2 \in \applast(P_1,m')$, the following holds for any alternative $a \in \ma_{m}$:
\begin{itemize}
    \item If $a$ is Dodgson winner in $P_1$, then $a$ is also Dodgson winner in $P_2$.
    \item The Dodgson score of $a$ in $P_1$ is equal to that in $P_2$.
\end{itemize}
\end{lemma}
The proof of Lemma \ref{lem:dodgson-toprobust} follows by definition and can be found in Appendix \ref{appendix:ds-toprobust}. 
Informally, Lemma \ref{lem:dodgson-toprobust} states that by appending alternatives at the bottom of each agent's preference order, the winner and score under the Dodgson rule are robust. 

Let $m = (2m_1n)^d = \polynomial(q)$, where $d$ is the constant defined in Assumption \ref{asmpt:general-hard}. We create a set of $m-m_1$ dummy alternatives called $D$. The total alternative set is set as $\ma_m = \ma_{m_1}\cup D$. Denote $P_1 = (R^1_i)_{i\in[n]}$. We define $P:= \applast(P_1,m-m_1) = (R_i)_{i \in [n]}$ by appending the dummy alternatives in $D$. We remark that by the definition of $\applast$, each ranking $R_i$ in $P$ is of the form 
\begin{align*}
    R_i = \ma_{m_1} \succ_{R_i} D
\end{align*}
where the order of $\ma_{m_1}$ in $R_i$ is the same as $R^1_i$.

Now we construct the parameter profile $P^\Theta$ based on $P$ such that each parameter corresponds to a preference order in $P$. According to Assumption \ref{asmpt:general-hard}, there exists $K = \lceil m^{\frac{1}{d}}\rceil \ge m_1$, $\ma'\subseteq \ma_m$, $R'\in \ml(\ma')$, and $\pi\in \Pi_m$, such that $|\ma'|=K$ and $\Pr_{R\sim \pi}(\text{Top}_K(R) = R') \ge 1-\frac{1}{K}$. 
Let \[R^*:= \ma' \succ_{R^*} (\ma_m\setminus \ma')\] where the order in $\ma'$ is the same as $R'$ and the order in $\ma_m\setminus \ma'$ is arbitrary. Denote the parameter corresponding to this specific distribution $\pi \in \Pi_m$ as $\theta$. For every $i \in [n]$, we can find a permutation $\sigma_i$ over $\ml(\ma_m)$ such that $\sigma_i(R^*) = R_i$. We then apply permutation $\sigma_i$ to the predefined distribution $\pi$ and get a new distribution $\sigma_i(\pi)$ which is also in $\Pi_m$ since $\mm_m$ is neutral,. Now we define the parameter profile $P^\Theta := (\theta_i)_{i\in[n]}$, where $\theta_i$ is the parameter corresponding to $\sigma_i(\pi)$. Since $K \ge m_1$, we have
\begin{align*}
    \Pr_{R \sim \pi_{\theta_i}}(\topgroup_{m_1}(R) = R^1_i) \ge 1-\frac{1}{K}
\end{align*}
and the construction of $P^\Theta$ can be done in polynomial time of $q$.

\noindent{\bf Step 2.~Use $\alg$ to solve \xc.} 
For a profile $P \in \ml(\ma_m)^n$, we denote $\topgroup_K(P)$ the collection of top-$K$ ranking of each preference order in $P$. We now prove that we can construct a coRP Algorithm for \xc{} based on \alg{}. 

\begin{algorithm}
\caption{Randomized Algorithm for \xc{}}
{\bf Input}: An \xc{} instance $(U,S)$ and $\alg$ for \ds{}
\begin{algorithmic}[1] 
\STATE Construct profile $P_1$ and parameter profile $P^{\Theta}$ according to Step 1.
\STATE Sample a profile $P'$ from $\vec \mm_m$ given $P^{\Theta}$.
\IF {$\topgroup_{m_1}(P') \ne P_1 $}
\STATE Return YES.
\ENDIF
\STATE Run $\alg$ on $(P',c,4q/3)$.
\IF {$\alg$ returns YES}
\STATE Return YES.
\ELSE
\STATE Return NO.
\ENDIF
\end{algorithmic}
\label{alg:rpalg}
\end{algorithm}

\begin{claim}\label{claim:reduction}
If $\topgroup_{m_1}(P') = P_1$, then $(P',c,4q/3)$ is a YES instance for \ds{} if and only if $(U,S)$ is a YES instance for \xc{}.
\end{claim}
\begin{proof}
We know that $P' \in \applast(P_1,m-m_1)$ by definition. According to Lemma \ref{lem:dodgson-toprobust}, we know that the Dodgson score of $c$ in $P'$ is the same as the Dodgson score of $c$ in $P_1$. Therefore, $(P',c,\frac{4q}{3})$ is a YES instance of \ds{} if and only if $(P_1,c,\frac{4q}{3})$ is a YES instance of \ds{}, which is also equivalent to $(U,S)$ is a YES instance by lemma \ref{lem:dodgson-reduction}. 
\end{proof}

Notice that sampling $P'$ from  $P^{\Theta}$ takes polynomial time because $\vec\mm$ is P-samplable (Assumption~\ref{asmpt:general-hard}). It follows that Algorithm~\ref{alg:rpalg} is a polynomial-time algorithm. Recall that A coRP algorithm always returns YES to YES instances, and returns NO with constant probability to NO instances.  Since Algorithm \ref{alg:rpalg} returns NO only if $\topgroup_{m_1}(P') = P_1$ and $(P',c,4q/3)$ is a NO instance, by claim \ref{claim:reduction} it is clear that if $(U,S)$ is a YES instance then Algorithm \ref{alg:rpalg} returns YES. Therefore, to prove that Algorithm~\ref{alg:rpalg} is an coRP algorithm it suffices to prove that if $(U,S)$ is a NO instance then Algorithm~\ref{alg:rpalg}  returns NO with constant probability.

\begin{claim}\label{claim:goodsample}
$\Pr\left(\topgroup_{m_1}(P') = P_1\right) \ge 1/2$.
\end{claim}
\begin{proof}
$P' = (R'_i)_{i \in [n]}$ is sampled from $P^{\Theta} = (\theta_i)_{i \in [n]}$. Recall that $m = (2m_1n)^d$ and $K = m^{\frac{1}{d}}$. Thus $K \ge m_1$ and we know that by construction in Step 1 and Assumption \ref{asmpt:general-hard} that for all $i \in [n]$,
\begin{align*}
    \Pr_{R_i' \sim \pi_{\theta_i}}\left(\topgroup_{m_1}(R_i') = R_i^1\right)\ge \Pr_{R_i' \sim \pi_{\theta_i}}\left(\topgroup_{K}(R_i') = R_i^1\right )
    \ge 1-\frac{1}{K} = 1 - \frac{1}{2m_1n}.
\end{align*}
Thus we can derive
\begin{align*}
   \Pr_{P' \sim P^\Theta}\left(\topgroup_{m_1}(P') = P_1\right) &
   \ge \prod_{i=1}^n\left(\Pr_{R_i' \sim \pi_{\theta_i}}\left(\topgroup_{K}(R_i') = R_i^1\right )\right)\\
   & \ge (1 - \frac{1}{2m_1n})^n \ge 1-\frac{1}{2m_1} \ge \frac{1}{2}. \qedhere
\end{align*}
\end{proof}
When $(U,S)$ is a NO instance of \xc{},  $\alg{}(P',c,\frac{4q}{3})$ returns NO with probability at least $1-\frac{1}{m}$ by definition \ref{dfn:smoothed-winner}. Note that Algorithm \ref{alg:rpalg} returns YES when $\topgroup_{m_1}(P') \ne P_1$ which happens with probability at most $\frac{1}{2}$ by Claim~\ref{claim:goodsample}. Therefore, for any profile $P' \in \ml(\ma_m)^n$ such that $\topgroup_{m_1}(P') = P_1\}$, $\alg(P',c,\frac{4q}{3})$ succeeds with probability at least $1-\frac{2}{m} \ge \frac{1}{3}$. According to Claim~\ref{claim:reduction} and ~\ref{claim:goodsample}, we know that Algorithm~\ref{alg:rpalg} returns NO for any NO instance with probability at least $\frac{1}{2}\times \frac{1}{3} = \frac{1}{6}$. This completes the proof.
\end{proof}
We prove a similar result for \ys{} and the proof can be found in Appendix \ref{appendix:young}.
\begin{theorem}[\bf Semi-random hardness of \ys{}]\label{thm:hard-ys} For any single-agent preference model $\vec\mm$ that satisfies Assumption~\ref{asmpt:general-hard}, there exists no polynomial-time algorithm for {\sc Semi-Random}-\ys{} under $\vec\mm$ unless {\sc NP}$=${\sc ZPP}.
\end{theorem}
\noindent{\em Proof sketch. }{\rm 
We first extend Lemma \ref{lem:dodgson-toprobust} for the Dodgson rule to the Young rule. With that in hand, the proof is then very similar to that of Theorem \ref{thm:hard-ds}. The main difference is now that we use the reduction in \citep{Caragiannis09:Approximability} to construct the profile in Step 1 and then a coRP algorithm for the NP-complete problem \xc{}, which leads to NP $=$ ZPP. \qed}

\subsection{Extension to Multi-Winner Voting Rules}
A multi-winner voting rule selects a winning \emph{$k$-committee}, which is a $k$-size subset of alternatives. We consider the \cc{} (CC) rule and the Monroe rule that assign each $k$-committee a score and choose the $k$-committee with the highest (respectively, lowest) score as the winner. Definitions of the two voting rules and their corresponding winner determination problems and the proof of the following theorem can be found in Appendix \ref{appendix:cc-monroe}.
We remark that winner determination under the CC rule and the Monroe rule are both NP-hard \citep{Procaccia2008:On-the-complexity,Lu2011BudgetedSC}.

\begin{theorem}[\bf Semi-random hardness of CC and Monroe]\label{thm:hard-cc-monroe} For any single-agent preference model $\vec\mm$ that satisfies Assumption~\ref{asmpt:general-hard}, there exists no polynomial-time algorithm for the semi-random version of the winner determination problems of the CC rule and the Monroe rule under $\vec\mm$ unless {\sc NP}$=${\sc ZPP}.
\end{theorem}
\noindent{\em Proof sketch. }{\rm 
We first prove counter parts of Lemma \ref{lem:dodgson-toprobust} for the CC rule and the Monroe rule. Then the proof follows the same idea in the proof of Theorem \ref{thm:hard-ds}, except that we use different reductions to construct the  profile in Step 1.  \qed}

%% file: ks-ds.tex
\section{\ks{} v.s. \ds{}}\label{sec:ks-ds}
In this section, we present to two results regarding the Kemeny and Dodgson rule under the $\vec{\alpha}$-IC model. In Theorem \ref{thm:hard-ks}, we show that {\sc Semi-Random}-\ks{} has no polynomial time algorithm under $(1-\frac{1}{m})$-IC unless NP = ZPP. In contrast, we provide an efficient algorithm for {\sc Semi-Random}-\ds{} under $(1-\frac{1}{m})$-IC when $n = \Omega(m^2\log^2 m)$ (Theorem \ref{thm:dodgson-easy-alphaIC}). The two results together provide an interesting separation of the semi-random complexity of winner determination under different NP-hard rules. 
\subsection{Semi-Random Hardness of \ks{}}
\ks{} is NP-complete and is easier than computing the Kemeny ranking, which is $\Theta_2^P$-complete \citep{Hemaspaandra05:Complexity}.
Thus hardness result for computing the Kemeny ranking \citep{Xia2021:The-Smoothed} does not imply the semi-random hardness of \ks{}.
Nevertheless, under the same assumption made in \citep{Xia2021:The-Smoothed}, we can prove the semi-random hardness of \ks{}. To better illustrate the separation of semi-random complexity between Kemeny and Dodgson, we state the result in a special case under the $\vec\alpha$-IC model first. The formal statement of the general assumption and theorem as well as its proof are defered to Section \ref{sec:proof of ks}.
\begin{theorem}\label{thm:hard-ks}
For any constant $d \ge 0$ and $\vec\alpha = (\alpha_m)_{m\ge 3}$ such that $\alpha_m \in [0,1-\frac{1}{m^d}]$ for any sufficiently large $m$, there exists no polynomial-time algorithm for {\sc Semi-Random}-\ks{} under $\vec\alpha$-IC unless {\sc NP} $=$ {\sc ZPP}.
\end{theorem}
\notshow{
\noindent{\em Proof sketch. }{\rm 
Assume {\sc Semi-Random}-\ks{} has polynomial time algorithm, we give a coRP algorithm for the NP-complete problem {\sc Feedback Arc Set} on Eulerian graphs (EFAS) \citep{Perrot2015:Feedback}. We use techniques developed in \citep{Xia2021:The-Smoothed} in construction of the parameter profile, but then we need a more involved analysis to give upper and lower bounds of the optimal Kemeny score in a semi-random profile to argue the algorithm is really a coRP algorithm.  \qed}}
Note that for $d \ge 0$, $(1-\frac{1}{m^d})$-IC is close to the average case in the sense that any distribution in $\Pi_m$ is only $O(\frac{1}{m^d})$ away from the uniform distribution in total variation distance. Therefore, Theorem \ref{thm:hard-ks} shows that \ks{} remains hard even for models that are close to the average case. 
\subsection{Semi-Random Easiness of \ds{}}
In contrast to the Kemeny rule, we prove that winner determination under the Dodgson rule is tractable under models close to the average case, i.e., $(1-\frac{1}{m})$-IC. We remark here that although $(1-\frac{1}{m})$-IC is close to IC, $(1-\frac{1}{m})$-IC may concentrate on a single ranking with probability as large as $\Theta(\frac{1}{m})$, while every ranking in IC has probability exactly $\frac{1}{m!}=o(\frac{1}{\exp(m)})$. 

Since 1-IC is equivalent to IC, the following theorem that works for any $\alpha \in [1-\frac{1}{m},1]$ thus generalizes previous results that only work for IC~\citep{McCabe-Dansted2008:Approximability,Homan2009:Guarantees}. 
\begin{theorem}[\bf Semi-random easiness of \ds{}]\label{thm:dodgson-easy-alphaIC}
For any $\vec\alpha = (\alpha_m)_{m\ge 3}$ such that $\alpha_m \in [1-\frac{1}{m},1]$ for sufficiently large $m$, there exists a polynomial-time algorithm for {\sc Semi-Random}-\ds{} under $\vec\alpha$-IC that succeeds with probability at least $1 - 2(m-1)\exp\left(-\frac{n}{72m^2} \right)$.
\end{theorem}
\begin{proof}
The algorithm runs the polynomial-time greedy algorithm, denoted as {\sc Greedy} in \citep{Homan2009:Guarantees} as a subroutine. Given $(P,a)$, the output of {\sc Greedy}$(P,a)$ belongs to $\mathbb{Z} \times \mathrm{(``definitely",``maybe")}$ such that if {\sc Greedy}$(P,a)$ outputs $(s,\mathrm{``definitely"})$, then $s$ is the Dodgson score of $a$ in $P$. Given \ds{} instance $(P,a,t)$, the algorithm runs {\sc Greedy}$(P,a)$ first. Then if {\sc Greedy}$(P,a)$ outputs $(s,\mathrm{``definitely"})$, the algorithm outputs YES or NO based on whether $s \le t$. Otherwise the algorithm declares failure. Therefore, it suffices to prove that when $P$ is generated from $\vec\alpha$-IC, {\sc Greedy}$(P,a)$ outputs with $\mathrm{``definitely"}$ with high probability. 

The following lemma, a simple extension of \cite[Theorem 4.1.1]{Homan2009:Guarantees}, gives a sufficient condition under which {\sc Greedy}$(P,a)$ outputs with $\mathrm{``definitely"}$. We introduce some new notations here. For two distinct alternatives $a,b$ and voter $i$, by $a \prec_i b$ we mean voter $i$ prefers $b$ to $a$. By $a \lessdot _i b$ we mean that not only voter $i$ prefers $b$ to $a$, but also there is no other alternative $c$ such that voter $i$ prefers $b$ to $c$ and prefers $c$ to $a$ i.e., $a \prec_i c \prec_i b$.

\begin{lemma}\label{lem:condition-definitely}
Given $P = (\prec_i)_{i\in [n]}$. For each alternative $a \in \ma_m$, if for all $b \in \ma_m\setminus\{a\}$ there exists $\beta > 0$ such that $\left|\{i\in [n]: a \prec_i b\}\right| \le \frac{n}{2}+\beta$ and $|\{i\in [n] : a \lessdot_i b\}|\ge \beta$ then {\sc Greedy}$(P,a)$ outputs with $\mathrm{``definitely"}$.
\end{lemma}

We give a sketch of the proof for Lemma~\ref{lem:condition-definitely} here. Recall that the Dodgson score of an alternative $a$ is the smallest number of exchanges between adjacent alternatives that makes $a$ a Condorcet winner. Now consider  alternative $b \ne a$ such that $a$ needs extra $\beta$ votes to defeat $b$. If $|\{i\in [n] : a \lessdot_i b\}|\ge \beta$, then $a$ defeats $b$ after exactly $\beta$ exchanges, which is also necessary. If this is the case for any alternative $b\ne a$, then we can decide in polynomial time the Dodgson score of $a$ with certainty.

\begin{claim}\label{Claim:dodgson concentration}
For any profile $P = (\succ_i)_{i\in [n]}$ generated from $\alpha_m$-IC, alternatives $a,b \in \ma_m$, and $\beta = (\frac{3}{4}-\frac{1}{2m})\frac{n}{m} > 0$, We have
\begin{itemize}
    \item $\Pr\left[|\{i\in [n] | a \prec_i b\}| > \frac{n}{2} + \beta\right] < \exp\left(-\frac{n}{72m^2} \right)$;
    \item $\Pr\left[|\{i\in [n] | a \lessdot_i b\}| < \beta\right] < \exp\left(-\frac{n}{72m^2} \right)$.
\end{itemize}
\end{claim}
\begin{proof}
Due to the space limit, we only prove the first inequality and leave the proof of the second inequality in Appendix \ref{appendix:dodgson-concentration}. 
We need the following technical lemma, which is a straightforward application of Hoeffding's inequality for bounded random variables, hence we omit the proof.
\begin{lemma}\label{lem: chernoff}
Let $X_1,\cdots,X_n$ be a sequence of mutually independent random variables. If there exist $q,p \in [0,1]$ such that $q \le p$ and for each $i \in \{1,\cdots,n\}$,
$$\Pr[X_i = 1-p] = q \text{ and }
    \Pr[X_i = -p] = 1-q,
$$
then for all $d  > 0$, we have $\Pr[\sum_{i=1}^n X_i > d] < e^{-2d^2/n}. $
\end{lemma}

Fix any $i \in [n]$. Denote $\pi_i \in \Pi_m$ the preference distribution of agent $i$. Since $\alpha_m \ge 1-\frac{1}{m}$, we know $\Pr_{\pi_i}[R] \ge \frac{m-1}{m\cdot m!}$ for any preference order $R \in \ml(\ma_m)$. Note that there are exactly $\frac{m!}{2}$ rankings in $\ml(\ma_m)$ such that $a$ is ranked above $b$. Therefore, we have
\begin{align*}
    \Pr[a \prec_i b]= 1- \Pr[b \prec_i a] \le 1-\frac{m!}{2}\cdot\frac{m-1}{m\cdot m!} = \frac{m+1}{2m}
\end{align*}
For each $i\in [n]$, define $X_i$ as
\[
    X_i = \begin{cases}
    \frac{m-1}{2m} & \mbox{if}~a \prec_i b\\
    -\frac{m+1}{2m} & \mbox{otherwise}
    \end{cases}
\]
It follows that $|\{i\in [n] | a \prec_i b\}| > \frac{n}{2} + \beta$ only if 
\begin{align*}
    \sum_{i=1}^n X_i &> \frac{m-1}{2m}\left(\frac{n}{2}+\beta\right)-\frac{m+1}{2m}\left(\frac{n}{2} -\beta\right) = \left(\frac{1}{4}-\frac{1}{2m}\right)\frac{n}{m} \ge\frac{n}{12m} \tag{$m \ge 3$}
\end{align*}
Note that $\Pr[X_i = \frac{m-1}{2m}] = \Pr[a \prec_i b] \le \frac{m+1}{2m}$. The claim follows by setting $d =\frac{n}{12m}$ and $p = \frac{m+1}{2m}$ in Lemma \ref{lem: chernoff}.
\end{proof}

Applying union bound for all $m-1$ alternatives in $\ma_m-\{a\}$ to Claim \ref{Claim:dodgson concentration}, we have
\begin{align*}
    &\Pr \bigg[\forall b \ne a, |\{i\in [n] : a \prec_i b\}| > \frac{n}{2} + \beta ~~\text{or}~~ |\{i\in [n]: a \lessdot_i b\}| < \beta \bigg] \\
    &\le 2(m-1)\exp\left(-\frac{n}{72m^2} \right)
\end{align*}
According to Lemma \ref{lem:condition-definitely}, with probability at least $1- 2(m-1)\exp\left(-\frac{n}{72m^2} \right)$, {\sc Greedy}$(P,a)$ outputs with $\mathrm{``definitely"}$. This completes the proof.
\end{proof}

According to Theorem \ref{thm:dodgson-easy-alphaIC}, we know that under $(1-\frac{1}{m})$-IC, {\sc Semi-Random}-\ds{} is in P when $n = \Omega(m^2\log^2 m)$. By Theorem \ref{thm:hard-ks}, {\sc Semi-Random}-\ks{} has no polynomial time algorithm under $(1-\frac{1}{m})$-IC unless NP = ZPP. The two results together provide an interesting separation of the semi-random complexity of winner determination under different NP-hard rules. 

\subsection{Proof of Theorem \ref{thm:hard-ks}}\label{sec:proof of ks}
We introduce some notations before the statement of assumption and the proof. For a profile $P \in \ml(\ma_m)^n$, its \emph{weighted majority graph \wmg(P)} is a weighted directed graph, and its vertices are represented by $\ma_m$. For any pair of alternatives $a,b\in \ma_m$, the weight on edge $a \rightarrow b$ is the number of agents that prefer $a$ to $b$ minus the number of agents that prefer $b$ to $a$. For a distribution $\pi$ over rankings, we define its weighted majority graph $\wmg(\pi)$ similarly: For any pair of alternatives $a,b\in \ma_m$, the weight on edge $a \rightarrow b$ is the probability that a ranking prefers $a$ to $b$ minus the probability that a ranking prefers $b$ to $a$. For each 3-cycle $a\rightarrow b \rightarrow c \rightarrow a$, its weight is defined as the sum of the weights on its three edges $a\rightarrow b$, $b\rightarrow c$, and $c \rightarrow a$. 

\begin{assumption}[\cite{Xia2021:The-Smoothed}]\label{asmpt:kemeny}$\vec\mm$ is P-samplable, neutral, and satisfies the following condition: there exist constants $k\ge 0$ and $A>0$ such that for any $m\ge 3$, there exist $\pi_{3c}\in \Pi_m$ such that $\wmg(\pi_{3c})$ has a 3-cycle $\gcyc$ with weight at least $\frac{A}{m^k}$
\end{assumption}

Assumption \ref{asmpt:kemeny} is weaker than Assumption \ref{asmpt:general-hard}. That's because in the distribution $\pi$ guaranteed by Assumption \ref{asmpt:general-hard}, the top-$K$ ranking remains unchanged with probability at least $1-\frac{1}{K}$, which implies that the 3-cycle formed by the top-3 alternatives has weight $\ge 1-\frac{2}{K}$ with $K = m^{\frac{1}{d}}$ for constant $d$. 
For $\alpha_m \in [0,1-\frac{1}{m^d}]$, the model $\alpha_m$-IC has a 3-cycle with weight at least $\mathcal{O}(\frac{1}{m^d})$ and thus also satisfies Assumption \ref{asmpt:kemeny}. We prove in Theorem~\ref{thn:hard-ks-general} the smoothed hardness of Kemeny under Assumption~\ref{asmpt:kemeny} which implies Theorem~\ref{thm:hard-ks}.

\begin{theorem}[\bf Smoothed Hardness of Kemeny]\label{thn:hard-ks-general}
For any single-agent preference model $\vec\mm$ that satisfies Assumption \ref{asmpt:kemeny}, there exists no polynomial-time algorithm for {\sc Semi-Random}-\ks{} unless {\sc NP}$=${\sc ZPP}. 
\end{theorem}

\begin{proof}
Suppose that {\sc Semi-Random}-\ks{} has a polynomial-time algorithm, denoted as \alg{}. We use it to construct a coRP algorithm for the NP-complete problem {\sc Eulerian Feedback Arc Set (EFAS)}~\cite{Perrot2015:Feedback}, which implies NP $=$ ZPP as discussed in the proof of Theorem \ref{thm:hard-ds}. An instance of \efas{} is denoted by $(G,t)$, where $t\in\mathbb N$ and $G$ is a directed unweighted Eulerian graph, which means that there exists a closed Eulerian walk that passes each edge exactly once. We are asked to decide whether $G$ can be made acyclic by removing no more than $t$ edges. 

Given a single-agent preference model, a {\em (fractional) parameter profile} $P^\Theta\in \Theta_m^n$ is a collection of $n>0$ parameters, where $n$ may not be an integer. Note that $P^\Theta$ naturally leads to a fractional preference profile, where the weight on each ranking represents its total weighted ``probability'' under all parameters in $P^\Theta$. We include an illustrating example of fractional parameter profile and fractional preference profile in Appendix~\ref{appendix:example}.

Let $(G = (V,E),t)$ be any \efas{} instance, where $|V| = m$.
\begin{claim}[\cite{Xia2021:The-Smoothed}]
We can construct a fractional preference profile $P^\Theta_G$ in polynomial time in $m$ such that there exists a constant $k$
\begin{itemize}
    \item $|P^\Theta_G| = \mathcal{O}(m^{k+2})$,
    \item $P_G^\Theta$ consists of $\mathcal{O}(m^5)$ types of parameters, 
    \item $\wmg(P_G^\Theta) = G$.
\end{itemize}
\end{claim}
 Let $K = 13+2k$, which means that $K>12$. We first define a parameter profile $P_G^{\Theta*}$ of $n=\Theta(m^{K})$ parameters that is approximately $\frac{m^{K}}{|P_G^{\Theta}|}$ copies of $P_G^\Theta$ up to $\mathcal{O}(m^5)$ in  $L_\infty$ error. Formally, let 
\begin{equation}\label{eq:pg}
P_G^{\Theta*} = \left\lfloor P_G^\Theta \cdot \dfrac{m^{K}}{|P_G^{\Theta}|}\right\rfloor
\end{equation}

Let $n=|P_G^{\Theta*}|$. Because the number of different types of parameters in $P_G^{\Theta*}$ is $\mathcal{O}(m^5)$, we have $n = m^{K}-\mathcal{O}(m^5)$, $\|\wmg(P_G^{\Theta*}) - \wmg(P_G^{\Theta} \cdot \frac{m^{K}}{|P_G^{\Theta}|}) \|_\infty  = \mathcal{O}(m^5)$, and $\|\wmg(P_G^{\Theta*}) - G \cdot \frac{m^{K}}{|P_G^{\Theta}|}) \|_\infty  = \mathcal{O}(m^5)$. Let $f(G,R)$ denote the number of backward arcs of linear order $R$ in a directed graph $G$. The following useful claim calculates the KT distance between $R$ and the parameter profile $P^\Theta_G \cdot \frac{m^K}{|P^\Theta_G|}$. The proof of Claim~\ref{claim:kt distance} can be found in Appendix~\ref{appendix:proof of kt distance}.

\begin{claim}
\label{claim:kt distance}
For any linear order $R \in \ml(\ma_m)$, the KT distance between $R$ and the fractional parameter profile $P^\Theta_G \cdot \frac{m^K}{|P^\Theta_G|}$ is $ \kt\left(P_G^\Theta \cdot \frac{m^{K}}{|P_G^{\Theta}|},R\right) = M + \frac{m^K}{|P^\Theta_G|} \cdot f(G,R)$, where $M =  \frac{m^K}{2}\left(\binom{m}{2}-\frac{|E|}{|P_G^\Theta|}\right)$.
\end{claim}

\begin{algorithm}
\caption{Algorithm for \efas{}.}
{\bf Input}: EFAS Instance $(G,t)$, $\alg$ 
\begin{algorithmic}[1] 
\STATE Compute a parameter profile $P_G^{\Theta*}$ according to (\ref{eq:pg}).
\STATE Sample a profile $P'$ from $\vec \mm_m$ given $P_G^{\Theta*}$.
\IF {$\|\wmg(P')-G_n\|_1 > \binom{m}{2}\cdot m^{\frac{K+1}{2}}$}
\STATE Return YES.
\ENDIF
\STATE Run $\alg$ on $\left(P', M+t\cdot\frac{m^K}{|P_G^\Theta|}+m^{k+10}\right)$.
\IF {\alg{} returns NO}
\STATE Return NO.
\ELSE
\STATE Return YES.
\ENDIF
\end{algorithmic}
\label{alg:rpalg-kemeny}
\end{algorithm}

We now prove that $\alg$ returns the correct answer to $(G,t)$ with probability at least $1-\exp(-\Omega(m))$. Let $G_n = G \cdot \frac{m^{K}}{|P_G^{\Theta}|}$. The following claim bounds the probability that $\wmg(P')$ is different from $G_n $ by more than $\Omega(m^{\frac{K+1}{2}})$.

\begin{claim}[\cite{Xia2021:The-Smoothed}]\label{claim:concentration}
$\Pr\left[\|\wmg(P') - G_n \|_1 >\binom{m}{2}\cdot m^{\frac{K+1}{2}}\right] <\exp(-\Omega(m))$.
\end{claim} 
\begin{claim}\label{claim:kemeny-reduction}
If $\|\wmg(P') - G_n \|_1 \le \binom{m}{2}\cdot m^{\frac{K+1}{2}}$, then $\left(P', M+t\cdot\frac{m^K}{2|P_G^\Theta|}+m^{k+10}\right)$ is a YES instance of \ks{} if and only if $(G,t)$ is a YES instance of \efas{}.
\end{claim}
\begin{proof} If $(G,t)$ is a YES instance of \efas{}, then there exists a linear order $R$ such that there are at most $t$ backward arcs in $G$ according to $R$. Considering $R$ as a ranking over alternatives, we have $ \kt\left(P^\Theta_G \cdot \frac{m^K}{|P^\Theta_G|},R\right)\le M+t\cdot\frac{m^K}{|P_G^\Theta|}$. By assumption we know $|\kt(P',R) - \kt(P_G^\Theta \cdot \dfrac{m^{K}}{|P_G^{\Theta}|},R)| =  \mathcal{O}(m^{\frac{K+5}{2}})$. Therefore, the kemeny score of ranking $R$ is at most 
\begin{align*}
    \kt(P',R) &\le \kt\left(P_G^\Theta \cdot \dfrac{m^{K}}{|P_G^{\Theta}|},R\right) + \mathcal{O}(m^{\frac{K+5}{2}})< M+t\cdot\frac{m^K}{|P_G^\Theta|}+m^{k+10},
\end{align*} 
which means $\left(P', M+t\cdot\frac{m^K}{|P_G^\Theta|}+m^{k+10}\right)$ is a YES instance.

If $(G,t)$ is a NO instance of \efas{}, then for any linear order $R$ of $|V|$, there are at least $t+1$ backward arcs in $G$ according to $R$. We have for any $R \in \ml(\ma_m)$, $ \kt\left(P_G^\Theta\cdot\frac{m^K}{|P_G^\Theta|},R\right)\ge M+(t+1)\cdot\frac{m^K}{|P_G^\Theta|}.$ Therefore, for any $R \in \ml(\ma_m)$, we have
\begin{align*}
    \kt(P',R) &\ge \kt\left(P_G^\Theta \cdot \dfrac{m^{K}}{|P_G^{\Theta}|},R\right) - \mathcal{O}(m^{\frac{K+5}{2}})\\
    &\ge M+t\cdot\frac{m^K}{|P_G^\Theta|}+\frac{m^K}{|P_G^\Theta|}-\mathcal{O}(m^{\frac{K+5}{2}})\\
    &= M+t\cdot\frac{m^K}{|P_G^\Theta|} + \Theta(m^{k+11}) - \mathcal{O}(m^{k+9})\\
    &> M+t\cdot\frac{m^K}{|P_G^\Theta|}+m^{k+10},
\end{align*}
which means $\left(P', M+t\cdot\frac{m^K}{|P_G^\Theta|}+m^{k+10}\right)$ is a NO instance of \ks{}.
\end{proof}

Note that Algorithm~\ref{alg:rpalg-kemeny} only returns NO in line 8, when $\|\wmg(P') - G_n \|_1 >\binom{m}{2}\cdot m^{\frac{K+1}{2}}$ and \alg{} returns NO. By Claim~\ref{claim:kemeny-reduction}, we know that Algorithm~\ref{alg:rpalg-kemeny} never returns NO for any YES instance of \efas{}, or equivalently, it always returns YES for YES instance. Since $\|\wmg(P') - G_n \|_1 \le\binom{m}{2}\cdot m^{\frac{K+1}{2}}$ holds with probability at least $1-\exp(-\Omega(m))$ and \alg{} returns with probability at least $1-\frac{1}{m}$, we know that Algorithm~\ref{alg:rpalg-kemeny} returns NO for NO instance of \efas{} with at least constant probability. This proves that \efas{} is in coRP and completes the proof.
\end{proof}

%% file: appendix.tex
\section{Proof of Lemma~\ref{lem:dodgson-reduction}}
\label{appendix:ds-reduction}
We first present the construction of profile $P_1$ and then we show why it satisfies the desired properties.
\begin{description}
    \item[Construction of $P_1$.] We first construct the set of alternatives $\ma_{m_1}$, which contains three type of alternatives. $\ma_{m_1}$ contains a \emph{critical alternative} $c$. For any $i\in [q]$ and element $u_i \in U$, $\ma_{m_1}$ contains two \emph{element alternatives} $a_i$ and $b_i$. We denote by $E = \{a_i,b_i| i\in \{1,\cdots,q\}\}$ the set of element alternatives. For any $j \in [s]$ and subset $S_j \in S$, $\ma_{m_1}$ contains one \emph{subset alternative} $s_j$. We denote by $H$ the set of subset alternatives.
    
    Now we construct the ranking profile $P_1$, which consists of the following three sub-profiles.
    \begin{description}
        \item[1. Swing Rankings] We create $s$ rankings for each member of $S$. For each subset $S_j = \{u_{j_1},u_{j_2},u_{j_3}\}$, denote $E_j = \{a_{j_1},a_{j_2},a_{j_3}\}$. Let $R_{S_j}$ be any ranking of the form $E_j \succ s_j \succ c \succ (E/\ E_j)\cup (H/\ s_j)$ where the order of alternatives in each part can be arbitrary. 
        We set $P_{1,1}$ be the profile containing $s$ swing rankings $R_{S_j}$ for all $j \in [s]$. It is easy to see $|P_{1,1}| = s$.

        The idea behind swing rankings $P_{1,1}$ is the following. Note that switching the special alternative $c$ up 1 position in $R_{S_j}$ gains 0 vote for $c$ against all element alternatives $E$; switching 2 times gains 1 vote; switching 3 times get 2 votes; switching 4 times get 3 votes. Thus, among the swing rankings $P_{1,1}$, any additional votes for $c$ over element alternatives in $E$ require $4/3$ switches per vote on the average. Moreover, to achieve $4/3$ switches per vote, $c$ must be switched 4 times to the very top in each switched ranking. 

        \item[2. Equalizing Rankings.] For $i \in [q]$, let $N_i = |\{S_j \in S | u_i \in S_j\}|$ be the number of subsets in $S$ that contains $u_i$. Let $N^* = \max\{N_1,N_2,\cdots,N_q\}$. Let $R_{u_i}$ be any linear order of the form $a_i \succ b_i \succ c \succ (E/\ \{a_i,b_i\}) \cup H$, where the order of the alternatives after $c$ can be arbitrary. We set $P_{1,2}$ to be the profile containing $N^{*} - N_i$ copies of $R_{u_i}$ for all $i \in [q]$. We also have $|P_{1,2}| \le \sum_{i \in [q]} N^*-N_i \le q N^* \le qs$.

        By adding equalizing rankings in $P_{1,2}$ to $P_{1,1}$, each alternative $a_i$ gets equal score in the pairwise competitions against $c$. Note that among equalizing rankings $P_{1,2}$, additional votes for $c$ over an element alternative $a_i$ require at least 2 switches per vote on the average.

        \item[3. Incremental Rankings.] Let $R_I$ be any ranking of the form $a_1 \succ \cdots \succ a_q \succ b_1 \succ \cdots \succ b_q \succ c \succ H $ where the order of the alternatives after $c$ can be arbitrary. We set $P_{1,3}$ the profile containing $N_I$ copies of $R_I$ such that $a_i$ will defeat $c$ by exactly 1 voter in $P_{1,1} \cup P_{1,2} \cup P_{1,3}$ for any $i \in [q]$. It is easy to see that $|P_{1,3}| \le |P_{1,1}| + |P_{1,2}| + 1 \le s+qs+1$.

        By adding incremental rankings $P_{1,3}$ to $P_{1,1}\cup P_{1,2}$, alternative $a_i$ defeats $c$ by exactly 1 vote for all $i \in [q]$. Besides, additional votes for $c$ over an element alternative $a_i$ among incremental rankings $P_{1,3}$ require at least 2 switches per vote in average. 
    \end{description}
    We set $P_1 = P_{1,1} \cup P_{1,2} \cup P_{1,3}$. According to the construction, we know that $|P_1| \le \sum_{i=1}^3 |P_{1,i}| \le 2(q+1)s+1 = \mathcal{O}(q^4)$. 
    \item[Reduction.] Recall that each element alternative $a_i$ wins exactly 1 vote against the critical alternative $c$ in $P_1$. Thus in order to make $c$ the Condorcet winner, $c$ must win against each $a_i$. However, this requires at least $4q/3$ switches and is achievable only if i) all switches are among swing rankings, and ii) each switched swing ranking move $c$ to the top of the preference by 4 switches. It is obvious from the construction of swing rankings that any collection of swing rankings that can elect $c$ by no more than $4q/3$ switches correspond to an exact 3-cover of $(U,S)$. Thus $(P_1,c,\frac{4q}{3})$ is a YES instance of \ds{} if and only if $(U,S)$ is a YES instance of \xc{}.
\end{description}
It is clear that the construction can be done in polynomial time in $q$. This completes the proof.

\section{Proof of Lemma \ref{lem:dodgson-toprobust}}\label{appendix:ds-toprobust}
Fix any profile $P_1 \in \ml(\ma_m)^n$, any integer $m'\ge 1$ and profile $P_2 \in \applast(P_1,m')$. Recall that the Dodgson score of alternative $a$ in profile $P$ is defined as the smallest number of sequential exchanges of adjacent alternatives in rankings of $P$ to make  $a$ the Condorcet winner. Therefore, the Dodgson score of $a$ does not depend on the order of alternatives that are less preferred than $a$ in each agent's preference. Thus the Dodgson score of $a$ in $P_1$ is equal to the Dodgson score of $a$ in $P_2 \in \applast(P_1,m')$. Besides, the Dodgson score of alternative $b \in \ma_{m+m'}\setminus \ma_{m}$ in $P_2$ is strictly higher than that of $a$ in $P_2$ since every agent prefers $a$ to $b$. Thus if $a$ is the Dodgson winner in $P_1$, $a$ is also the Dodgson winner in $P_2$.

\notshow{
\section{Proof of Lemma \ref{lem:dodgson-reduction}}\label{appendix:ds-reduction}
\paragraph{Construction of $P_1$.}
We first construct the set $\ma_{m_1}$ of alternatives, where $m_1 = 1+2q+s = \mathcal{O}(q^3)$. $\ma_{m_1}$ contains a critical alternative $c$. For any $i\in [q]$ and element $u_i \in U$, $\ma_{m_1}$ contains two element alternatives $a_i$ and $b_i$. We denote by $E = \{a_i,b_i| i\in \{1,\cdots,q\}\}$ the set of element alternatives. For any $j \in [s]$ and subset $S_j \in S$, $\ma_{m_1}$ contains one subset alternative $s_j$. We denote by $H$ the set of subset alternatives.

Now we construct the ranking profile $P_1$, which consists of the following three sub-profiles.

\noindent{\bf 1. Swing Rankings.}
We create $s$ rankings for each member of $S$. For each subset $S_j = \{u_{j_1},u_{j_2},u_{j_3}\}$, denote $E_j = \{a_{j_1},a_{j_2},a_{j_3}\}$. Let $R_{S_j}$ be any ranking of the form $E_j \succ s_j \succ c \succ (E/\ E_j)\cup (H/\ s_j)$ where the order of alternatives in each part can be arbitrary. 
We set $P_{1,1}$ be the profile containing $s$ swing rankings $R_{S_j}$ for all $j \in [s]$. It is easy to see $|P_{1,1}| = s$.

The idea behind swing rankings $P_{1,1}$ is the following. Note that switching the special alternative $c$ up 1 position in $R_{S_j}$ gains 0 vote for $c$ against all element alternatives $E$; switching 2 times gains 1 vote; switching 3 times get 2 votes; switching 4 times get 3 votes. Thus, among the swing rankings $P_1,1$, any additional votes for $c$ over element alternatives in $E$ require 3/4 switches per vote on the average. Moreover, to achieve 3/4 switches per vote, $c$ must be switched 4 times to the very top in each switched ranking.  

\noindent{\bf 2. Equalizing Rankings.}
For $i \in [q]$, let $N_i = |\{S_j \in S | u_i \in S_j\}|$ be the number of subsets in $S$ that contains $u_i$. Let $N^* = \max\{N_1,N_2,\cdots,N_q\}$. Let $R_{u_i}$ be any linear order of the form $a_i \succ b_i \succ c \succ (E/\ \{a_i,b_i\}) \cup H$, where the order of the alternatives after $c$ can be arbitrary. We set $P_{1,2}$ be the profile containing $N^{*} - N_i$ copies of $R_{u_i}$ for all $i \in [q]$. We also have $|P_{1,2}| \le \sum_{i \in [q]} N^*-N_i \le q N^* \le qs$.

By adding equalizing rankings in $P_{1,2}$ to $P_{1,1}$, each alternative $a_i$ gets equal score in the pairwise competitions against $c$. Note that among equalizing rankings $P_{1,2}$, additional votes for $c$ over element alternatives $E$ require at least 2 switches per vote on the average.

\noindent{\bf 3. Incremental Parameters.} 
Let $R_I$ be any ranking of the form $a_1 \succ \cdots \succ a_q \succ b_1 \succ \cdots \succ b_q \succ c \succ H $ where the order of the alternatives after $c$ can be arbitrary. We set $P_{1,3}$ the profile containing $N_I$ copies of $R_I$ such that $a_i$ will defeat $c$ by exactly 1 voter in $P_{1,1} \cup P_{1,2} \cup P_{1,3}$ for any $i \in [q]$. It is easy to see that $|P_{1,3}| \le |P_{1,1}| + |P_{1,2}| + 1 \le s+qs+1$.

By adding incremental rankings $P_{1,3}$ to $P_{1,1}\cup P_{1,2}$, alternative $a_i$ defeats $c$ by exactly 1 vote for all $i \in [q]$. Besides, additional votes for $c$ over element alternatives $E$ among incremental rankings $
P_{1,3}$ require at least 2 switches per vote in average. 

We set $P = P_{1,1} \cup P_{1,2} \cup P_{1,3}$. According to the construction, we know that $|P| \le \sum_{i=1}^3 |P_{1,i}| \le 2(q+1)s+1 = \mathcal{O}(q^4)$.

\paragraph{Reduction.} Recall that each element alternative $a_i$ wins exactly 1 vote against the critical alternative $c$ in $P_1$. Thus in order to make $c$ the Condorcet winner, $c$ must win against each $a_i$. However, this requires at least $3q/4$ switches and is achievable only if i) all switches are among swing rankings, and ii) each switched swing ranking move $c$ to the top of the preference by 4 switches. It is obvious from the construction of swing rankings that any collection of swing rankings that can elect $c$ by no more than $4q/3$ switches correspond to an exact 3-cover of $(U,S)$. This completes the proof.
}
\section{Proof of Theorem \ref{thm:hard-ys}}\label{appendix:young}
Note that we only need to provide the following analog results regarding the Young rule to make the proof of the semi-random hardness of \ds{} (Theorem \ref{thm:hard-ds}) also work for \ys{}.

\begin{definition}[\bf {\sc Semi-Random}-\ys{}]
Fix a series of single-agent preference models $\vec\mm$. Given alternative $a \in \ma_m$, $t \in \mathbb{N}$ and a semi-random profile $P$ drawn from $\mm_m$, we are asked to decide whether the Young score of $a$ is at least $t$, with probability at least $1-\frac{1}{m}$.
\end{definition}

\begin{lemma}\label{lem:young-toprobust}
For any profile $P_1 \in \ml(\ma_m)^n$, any integer $m'\ge 1$ and profile $P_2 \in \applast(P_1,m')$, the following holds for any alternative $a \in \ma_{m}$:
\begin{itemize}
    \item If $a$ is Young winner in $P_1$, then $a$ is also Young winner in $P_2$.
    \item The Young score of $a$ in $P_1$ is equal to that in $P_2$.
\end{itemize}
\end{lemma}
\begin{proof}
Fix any profile $P_1 \in \ml(\ma_m)^n$, any integer $m'\ge 1$ and profile $P_2 \in \applast(P_1,m')$. Recall that the Young score of alternative $a$ in profile $P$ is defined as the size of the largest subset of $P$ to make $a$ the Condorcet winner. Therefore, the Young score of $a$ does not depend on the order of alternatives that are less preferred than $a$ in each agent's preference. Thus the Young score of $a$ in $P_1$ is equal to the Young score of $a$ in $P_2 \in \applast(P_1,m')$. Besides, the Young score of alternative $b \in \ma_{m+m'}\setminus \ma_{m}$ in $P_2$ is exactly $n$ where the Young score of $a$ is at most $n$. Thus if $a$ is the Young winner in $P_1$, $a$ is also the Young winner in $P_2$.
\end{proof}
The following lemma directly follows from the reduction in \citep{Caragiannis09:Approximability} and we omit the proof.
\begin{lemma}\label{lem:young-reduction}
We can construct a profile $P_1 \in \ml(\ma_{m_1})^n$ with $m_1 = \mathrm{poly}(n)$ and $n = \mathrm{poly(n)}$, and an alternative $c$
 such that $(P_1,c,1)$ is a YES instance of \ys{} if and only if $(U,S)$ is a YES instance of \xc{}. The construction can be done in polynomial time in $q$.
\end{lemma}

\section{Extension to Multi-Winner Voting Rules}\label{appendix:cc-monroe}
\subsection{The $\cc$ rule and the Monroe rule.}  \

A \emph{positional scoring function} (PSF) is a function $\alpha^m: [m] \rightarrow \mathbb{Z}$. A PSF $\alpha^m$ is a \emph{decreasing positional scoring function} (DPSF) if for each $i,j \in [m]$, if $i < j$ then $\alpha^m(i) > \alpha^m(j)$.
Denote a family of DPSFs $(\alpha^m)_{m\ge3}$, where $\alpha^m$ is a DPSF on $[m]$, such that $\alpha^{m+1}(i) = \alpha^m(i)$ holds for all $m \ge 3$ and $i \in [m]$\footnote{Note that any family of PSFs would satisfy this constraint by consistently modifying the score on each position. This modification does not affect the complexity of winner determination problems. }. 
For each alternative $a$, denote $\pos_i(a)$ its position in agent $i$'s preference. 
Agent $i$'s satisfaction for $a$ is given by $\alpha^m(\pos_i(a))$. 


We now define the $\cc$ (CC) rule and the Monroe rule. 
Given a committee $C$, we call the function $\Phi_C : \mn \rightarrow C$ an {\em assignment function} for $C$, where $\mn$ is the set of agents. A {\em Monroe assignment function} should satisfy another constraint that the number of agents that assigned to each alternative is approximately equal:  $ \lfloor \frac{|\mn|}{|C|} \rfloor \le |\Phi_c^{-1}(a)| \le \lceil \frac{|\mn|}{|C|}\rceil$ holds for each alternative $a \in C$. Given a DPSF $\alpha$, committee $C \subseteq \ma_m$, and an (Monroe) assignment function $\Phi_C$, we use the following two functions to aggregate individual satisfaction:
\begin{align*}
  &\mathcal{I}_{\text{sum}}^{\alpha}(\Phi_C) = \sum_{i=1}^n \alpha^m(\pos_i(\Phi_C(i))),\\
  &\mathcal{I}_{\text{min}}^{\alpha}(\Phi_C) = \min_{i\in[n]} \alpha^m(\pos_i(\Phi_C(i))).
\end{align*}

We use the first function, $\mathcal{I}_{\text{sum}}^{\alpha}(\Phi_C)$ in the utilitarian framework. The score assigned to committee $C$ in the CC rule, denote as $\text{Utilitarian-Score}^{\alpha}_{\text{CC}}(C)$ is the maximum of $\mathcal{I}_{\text{sum}}^{\alpha}(\Phi_C)$ over all assignment functions. Similarly, the score assigned to committee $C$ in the Monroe rule, denote as $\text{Utilitarian-Score}^{\alpha}_{\text{M}}(C)$ is the maximum of $\mathcal{I}_{\text{sum}}^{\alpha}(\Phi_C)$ over all Monroe assignment functions. In the egalitarian framework, we use the second function in the definition of $\text{Egalitarian-Score}^{\alpha}_{\text{CC}}(C)$ and $\text{Egalitarian-Score}^{\alpha}_{\text{M}}(C)$.

The corresponding decision problems of the CC rule and the Monroe rule are defined below, where we are asked to decide whether there exists a $k$-committee whose score is at least a given threshold. Note that hardness of this problem implies hardness of finding the winners. 
\begin{definition}[\bf Winner determination problems of the CC rule and the Monroe rule]
Let $\alpha$ be any family of DPSFs. We are given $P \in \ml(\ma_m)^n, k \in [m], t \in \mathbb{Z}$. In $\alpha\satcc{}$ under the utilitarian framework, we are asked whether there exists a $k$-committee $C$ such that  $\textnormal{Utilitarian-Score}^{\alpha}_{\text{M}}(C)$ is at least $t$. In $\alpha\satcc{}$ under the egalitarian framework, we are asked whether there exists a $k$-committee $C$ such that  $\textnormal{Egalitarian-Score}^{\alpha}_{\text{M}}(C)$ is at least $t$. 
The problem $\alpha\satm{}$ is defined similarly.
\end{definition}
A family of functions $s = \{s^k_{m,n} : \ml(\ma_m)^n \times \ma_m^k  \rightarrow \mathbb{Z},  n,k\ge 1, m\ge \max\{3,k\}\}$ is called a \emph{score function}. 
For each preference profile $P \in \ml(\ma_m)^n$, a \emph{score-based voting rule} $r_s$ assigns a score to each $k$-committee according to $s^k_{m,n}$, and chooses the winners to be the set of $k$-committees with the highest (or the lowest score). 
called the $(r_s,P,k)$-winners. Note that the CC rule and the Monroe rule defined above are both score-based voting rules and their winner determination problems are of the following form: Given a voting profile $P \in \ml(\ma_m)^n$, a score function $s_{m,n}^k$, and $t \in \mathbb{Z}$, decide if there exists a $k$-committee $C \subseteq \ma_m$ such that $s_{m,n}^k(C)$ is at least $t$. So we also use \winner{$r_s$} to refer the winner determination problem under the voting rule $r_s$.

\subsection{Proof of Theorem \ref{thm:hard-cc-monroe}}
Note that we only need to provide the following results  to make the proof of the semi-random hardness of \ds{} (Theorem \ref{thm:hard-ds}) also work for the CC rule and the Monroe rule.
\begin{lemma}\label{lem:cc-monroe-toprobust}
Let $r_s$ be the CC rule or the Monroe rule. For any profile $P_1 \in \ml(\ma_m)^n$, any integer $m'\ge 1$ and profile $P_2 \in \applast(P_1,m')$, the following holds for any $k$-committee $C \subseteq \ma_{m}$:
\begin{itemize}
    \item If $C$ is $r_s$ winner in $P_1$, then $a$ is also $r_s$ winner in $P_2$.
    \item The $r_s$ score of $a$ in $P_1$ is equal to that in $P_2$.
\end{itemize}
\end{lemma}
\begin{proof}
We only prove for the CC rule and the utilitarian framework. The proof for other cases such as the Monroe rule and the egalitarian framework is similar. 

According to the definition of family of DPSF, we have $\alpha^{m}(i) = \alpha^{m+1}(i) = \cdots = \alpha^{m+m'}(i)$. Since the agent set does not change, the optimal assignment function $\Phi_C$ is the same for $P_1$ and $P_2$. It follows that 
\begin{align*}
   s^k_{m,n}(P_1,C)&= \sum_{i=1}^n \alpha^{m}(\pos_i(\Phi_c(i)))\\
    &= \sum_{i=1}^n \alpha^{m+m'}(\pos_i(\Phi_c(i)))\\
    &= s^k_{m+m',n}(P_2,C)
\end{align*}

Since $\alpha$ is a family of DPSF, any alternative in $\ma_{m+m'} - \ma_{m}$ has strictly lower score than any alternative in $\ma_{m_1}$ for any agent. 
It follows that $(r_s,P_1,k)$-winner $\subseteq$ $(r_s, P_2,k)$-winner. This completes the proof.
\end{proof}
The following lemma follows by reductions in \citep{Procaccia2008:On-the-complexity}.
\begin{lemma}\label{lem:cc-monroe-reduction}
Let $r_s$ be the CC rule or the Monroe rule. 
We can construct a profile $P_1 \in \ml(\ma_{m_1})^n$ with $m_1 = \mathrm{poly}(n)$ and $n = \mathrm{poly(n)}$, $k \in [m_1]$, and $t \in \mathbb{Z}$ such that 
such that $(P_1,k,t)$ is a YES instance of \winner{$r_s$} if and only if $(U,S)$ is a YES instance of \xc{}. The construction can be done in polynomial time in $q$.
\end{lemma}
Note that reduction from any NP-complete problem works here and there is nothing special with \xc{}. 

\notshow{
\section{Semi-Random Hardness of \ks{}: Proof of Theorem \ref{thm:hard-ks}}\label{appendix:kemeny}
We introduce some notations before the statement of assumption and the proof. For a profile $P \in \ml(\ma_m)^n$, its \emph{weighted majority graph \wmg(P)} is a weighted directed graph, and its vertices are represented by $\ma_m$. For any pair of alternatives $a,b\in \ma_m$, the weight on edge $a \rightarrow b$ is the number of agents that prefer $a$ to $b$ minus the number of agents that prefer $b$ to $a$. For each 3-cycle $a\rightarrow b \rightarrow c \rightarrow a$, its weight is defined as the sum of the weights on its three edges $a\rightarrow b$, $b\rightarrow c$, and $c \rightarrow a$. 
\begin{assumption}[\cite{Xia2021:The-Smoothed}]\label{asmpt:kemeny}$\vec\mm$ is P-samplable, neutral, and satisfies the following condition: there exist constants $k\ge 0$ and $A>0$ such that for any $m\ge 3$, there exist $\pi_{3c}\in \Pi_m$ such that $\wmg(\pi_{3c})$ has a 3-cycle $\gcyc$ with weight at least $\frac{A}{m^k}$
\end{assumption}

Assumption \ref{asmpt:kemeny} is weaker than Assumption \ref{asmpt:general-hard}. That's because in $\pi$ guaranteed by Assumption \ref{asmpt:general-hard}, the top-$K$ ranking remains unchanged with probability at least $1-\frac{1}{K}$, which implies that the 3-cycle formed by the top-3 alternatives has non-negligible weight at least $1-\frac{2}{K}$ with $K = m^{\frac{1}{d}}$ for constant $d$. 

Besides, $\alpha_m$-IC such that $\alpha_m \in [0,1-\frac{1}{m^d}]$ has a cycle of weight at least $\mathcal{O}(\frac{1}{m^d})$ and thus also satisfies Assumption \ref{asmpt:kemeny}. Therefore, Theorem \ref{thm:hard-ks} is a direct corollary of the following theorem.

\begin{definition}[\bf {\sc Semi-Random}-\ks{}]
Fix a series of single-agent preference models $\vec\mm$. Given $t \in \mathbb{N}$ and a semi-random profile $P$ drawn from $\mm_m$, we are asked to decide whether there exists an alternative whose Kemeny score of $a$ is at most $t$, with probability at least $1-\frac{1}{m}$.
\end{definition}
\begin{theorem}[\bf Smoothed Hardness of Kemeny]\label{thn:hard-ks-general}
For any single-agent preference model $\vec\mm$ that satisfies Assumption \ref{asmpt:kemeny}, there exists no polynomial-time algorithm for {\sc Semi-Random}-\ks{} unless {\sc NP}$=${\sc ZPP}. 
\end{theorem}

\begin{proof}
Suppose that {\sc Semi-Random}-\ks{} has a polynomial-time algorithm, denoted as \alg{}. We use it to construct a coRP algorithm for the NP-complete problem {\sc Eulerian Feedback Arc Set (EFAS)}~\cite{Perrot2015:Feedback}, which implies NP $=$ ZPP as discussed in the proof of Theorem \ref{thm:hard-ds}. An instance of \efas{} is denoted by $(G,t)$, where $t\in\mathbb N$ and $G$ is a directed unweighted Eulerian graph, which means that there exists a closed Eulerian walk that passes each edge exactly once. We are asked to decide whether $G$ can be made acyclic by removing no more than $t$ edges. 

Given a single-agent preference model, a {\em (fractional) parameter profile} $P^\Theta\in \Theta_m^n$ is a collection of $n>0$ parameters, where $n$ may not be an integer. Note that $P^\Theta$ naturally leads to a fractional preference profile, where the weight on each ranking represents its total weighted ``probability'' under all parameters in $P^\Theta$.

Let $(G = (V,E),t)$ be any \efas{} instance, where $|V| = m$.
\begin{claim}[\cite{Xia2021:The-Smoothed}]
We can construct a fractional preference profile $P_\Theta^G$ in polynomial time in $m$ such that 
\begin{itemize}
    \item $|P^\Theta_G| = \mathcal{O}(m^{k+2})$,
    \item $P_G^\Theta$ consists of $\mathcal{O}(m^5)$ types of parameters, 
    \item $\wmg(P_G^\Theta) = G$.
\end{itemize}
\end{claim}
 Let $K = 13+2k$, which means that $K>12$. We first define a parameter profile $P_G^{\Theta*}$ of $n=\Theta(m^{K})$ parameters that is approximately $\frac{m^{K}}{|P_G^{\Theta}|}$ copies of $P_G^\Theta$ up to $\mathcal{O}(m^5)$ in  $L_\infty$ error. Formally, let 
\begin{equation}\label{eq:pg}
P_G^{\Theta*} = \left\lfloor P_G^\Theta \cdot \dfrac{m^{K}}{|P_G^{\Theta}|}\right\rfloor
\end{equation}

Let $n=|P_G^{\Theta*}|$. Because the number of different types of parameters in $P_G^{\Theta*}$ is $\mathcal{O}(m^5)$, we have $n = m^{K}-\mathcal{O}(m^5)$, $\|\wmg(P_G^{\Theta*}) - \wmg(P_G^{\Theta} \cdot \frac{m^{K}}{|P_G^{\Theta}|}) \|_\infty  = \mathcal{O}(m^5)$, and $\|\wmg(P_G^{\Theta*}) - G \cdot \frac{m^{K}}{|P_G^{\Theta}|}) \|_\infty  = \mathcal{O}(m^5)$. Let $f(G,R)$ denote the number of backward arcs of linear order $R$ in graph $G$. Since $\wmg(P_G^\Theta) = G = (V,E)$, we have $ kt\left(P_G^\Theta \cdot \dfrac{m^{K}}{|P_G^{\Theta}|},R\right) = M + t\cdot f(G,R)$ where $M = |E|\cdot\left(\frac{m^K}{2}-\frac{m^K}{2|P_G^\Theta|}+(\binom{m}{2}-|E|\right)\cdot\frac{m^K}{2}$ =$ \frac{m^K}{2}\left(\binom{m}{2}-\frac{|E|}{|P_G^\Theta|}\right)$.

\begin{algorithm}
\caption{Algorithm for \efas{}.}
{\bf Input}: \efas{} instance $(G,t)$, $\alg$ for $\vec\mm$-Smooothed-\ks{}
\begin{algorithmic}[1] 
\STATE Compute a parameter profile $P_G^{\Theta*}$ according to (\ref{eq:pg}).
\STATE Sample a profile $P'$ from $\vec \mm_m$ given $P_G^{\Theta*}$.
\IF {$\|\wmg(P')-G_n\|_1 > \binom{m}{2}\cdot m^{\frac{K+1}{2}}$}
\STATE Return YES.
\ENDIF
\STATE Run $\alg$ on $\left(P', M+t\cdot\frac{m^K}{2|P_G^\Theta|}+m^{k+10}\right)$.
\IF {\alg{} returns NO}
\STATE Return NO.
\ELSE
\STATE Return YES.
\ENDIF
\end{algorithmic}
\label{alg:rpalg-kemeny}
\end{algorithm}

We now prove that $\alg$ returns the correct answer to $(G,t)$ with probability at least $1-\exp(-\Omega(m))$. Let $G_n = G \cdot \frac{m^{K}}{|P_G^{\Theta}|}$. The following claim bounds the probability that $\wmg(P')$ is different from $G_n $ by more than $\Omega(m^{\frac{K+1}{2}})$.

\begin{claim}[\cite{Xia2021:The-Smoothed}]\label{claim:concentration}
$\Pr\left[\|\wmg(P') - G_n \|_1 >\binom{m}{2}\cdot m^{\frac{K+1}{2}}\right] <\exp(-\Omega(m))$.
\end{claim}

\begin{claim}\label{claim:kemeny-reduction}
If $\|\wmg(P') - G_n \|_1 \le \binom{m}{2}\cdot m^{\frac{K+1}{2}}$, then $\left(P', M+t\cdot\frac{m^K}{2|P_G^\Theta|}+m^{k+10}\right)$ is a YES instance of \ks{} if and only if $(G,t)$ is a YES instance of \efas{}.
\end{claim}
\begin{proof} If $(G,t)$ is a YES instance of \efas{}, then there exists a linear order $R$ of $|V|$ such that there are at most $t$ backward arcs in $G$ according to $R$. Considering $R$ as a ranking over alternatives, we have $ kt\left(P^\Theta_G \cdot \frac{m^K}{|P^\Theta_G|},R\right)\le M+t\cdot\frac{m^K}{2|P_G^\Theta|}$. According to Claim~\ref{claim:concentration}, for any ranking $R$, by assumption we know $|kt(P',R) - kt(P_G^\Theta \cdot \dfrac{m^{K}}{|P_G^{\Theta}|},R)| =  \mathcal{O}(m^{\frac{K+5}{2}})$. Therefore, the kemeny score of any ranking $R$ is at most 
\begin{align*}
    kt(P',R) &\le kt\left(P_G^\Theta \cdot \dfrac{m^{K}}{|P_G^{\Theta}|},R\right) + \mathcal{O}(m^{\frac{K+5}{2}}) \\
    &< M+t\cdot\frac{m^K}{2|P_G\Theta|}+m^{k+10},
\end{align*} 
which means $\left(P', M+t\cdot\frac{m^K}{2|P_G^\Theta|}+m^{k+10}\right)$ is a YES instance.

If $(G,t)$ is a NO instance of \efas{}, then for any linear order $R$ of $|V|$, there are at least $t+1$ backward arcs in $G$ according to $R$. We have for any $R \in \ml(\ma_m)$, $ kt\left(P_G^\Theta\cdot\frac{m^K}{|P_G^\Theta|},R\right)\ge M+(t+1)\cdot\frac{m^K}{2|P_G^\Theta|}.$ Therefore, for any $R \in \ml(\ma_m)$, we have
\begin{align*}
    kt(P',R) &\ge kt\left(P_G^\Theta \cdot \dfrac{m^{K}}{|P_G^{\Theta}|},R\right) - \mathcal{O}(m^{\frac{K+5}{2}})\\
    &\ge M+t\cdot\frac{m^K}{2|P_G^\Theta|}+\frac{m^K}{2|P_G^\Theta|}-\mathcal{O}(m^{\frac{K+5}{2}})\\
    &= M+t\cdot\frac{m^K}{2|P_G^\Theta|} + \Theta(m^{k+11}) - \mathcal{O}(m^{k+9})\\
    &> M+t\cdot\frac{m^K}{2|P_G^\Theta|}+m^{k+10},
\end{align*}
which means $\left(P', M+t\cdot\frac{m^K}{2|P_G^\Theta|}+m^{k+10}\right)$ is a NO instance of \ks{}.
\end{proof}

Note that Algorithm~\ref{alg:rpalg-kemeny} only returns NO in line 8, when $\|\wmg(P') - G_n \|_1 >\binom{m}{2}\cdot m^{\frac{K+1}{2}}$ and \alg{} returns NO. By Claim~\ref{claim:kemeny-reduction}, we know that Algorithm~\ref{alg:rpalg-kemeny} never returns NO for any YES instance of \efas{}, or equivalently, it always returns YES for YES instance. Since $\|\wmg(P') - G_n \|_1 \le\binom{m}{2}\cdot m^{\frac{K+1}{2}}$ holds with probability at least $1-\exp(-\Omega(m))$ and \alg{} returns with probability at least $1-\frac{1}{m}$, we know that Algorithm~\ref{alg:rpalg-kemeny} returns NO for NO instance of \efas{} with at least constant probability. This proves that \efas{} is in coRP and completes the proof.
\end{proof}
}
\section{Proof of Claim \ref{Claim:dodgson concentration}}\label{appendix:dodgson-concentration}
Fix any $i \in [n]$. Denote $\pi_i \in \Pi_m$ the preference distribution of agent $i$. Since $\alpha_m \ge 1-\frac{1}{m}$, we know that $\Pr_{\pi_i}[R] \ge \frac{m-1}{m\cdot m!}$ for any preference order $R \in \ml(\ma_m)$. Therefore
\[
    \Pr[a \lessdot_i b] \ge (m-1)!\cdot\frac{m-1}{m\cdot m!} = \frac{m-1}{m^2}.
\]
For each $i\in \{1,\cdots,n\}$, define 
$$
    X_i = \begin{cases}
    \frac{m-1}{m^2} & \mbox{if}~a \centernot\lessdot_i b,\\
    \frac{m-1}{m^2}-1 & \mbox{otherwise}.
    \end{cases}
$$
Then $|\{i\in \{1,\cdots,n\} | a \lessdot_i b\}| < \beta$ if and only if $|\{i\in \{1,\cdots,n\} | a \centernot\lessdot_i b\}| > n - \beta$, which happens only if 
\begin{align*}
    \sum_{i=1}^n X_i >& \frac{m-1}{m^2}\left(n-\beta\right)+\left(\frac{m-1}{m^2}-1\right)\beta\\
    =& \left(\frac{1}{4}-\frac{1}{2m}\right)\frac{n}{m}
    \\&\ge\frac{n}{12m}.
\end{align*}
Note that $\Pr[X_i = \frac{m-1}{m^2}] = 1-\Pr[a \lessdot b] \le 1 - \frac{m-1}{m^2}$. Setting $d =\frac{n}{12m}$ and $p =1-\frac{m-1}{m^2}$ in Lemma \ref{lem: chernoff} yields the desired result. 

\section{An Example of Fractional Parameter Profile and Fractional Preference Profile}
\label{appendix:example}
Let $\{a_1, a_2, a_3\}$ be three alternatives. For simplicity, we denote 
\begin{align*}
    R(i,j,k):= a_i \succ a_j \succ a_k, \quad  \forall \{i,j,k\} = \{1,2,3\}
\end{align*}
We define two parameters $\theta_1$ and $\theta_2$, each of which induced a distribution over $\ml(\{a_1, a_2, a_3\})$. The two distribution $\pi_{\theta_1}$ and $\pi_{\theta_2}$ are defined as
\begin{align*}
    \Pr_{r \sim \pi_{\theta_1}}[r = R(1,2,3)] = \Pr_{r \sim \pi_{\theta_1}}[r = R(1,3,2)] =  \frac{1}{2} \\
    \Pr_{r \sim \pi_{\theta_2}}[r = R(1,2,3)] = \Pr_{r \sim \pi_{\theta_2}}[r = R(3,2,1)] =  \frac{1}{2}
\end{align*}
A \emph{fractional parameter profile} is collection of parameters with possibly non-integer weights. Similarly, a \emph{fractional preference profile} is collection of preference orders with possibly non-integer weights. For example, $(\alpha \theta_1, \beta \theta_2)$ is a fractional parameter profile for any $\alpha,\beta \ge 0$ and its induced fractional preference profile is 
\begin{align*}
    \{ \frac{\alpha + \beta}{2} R(1,2,3), \frac{\alpha}{2} R(1,3,2), \frac{\beta}{2} R(3,2,1) \}.
\end{align*}

\section{Proof of Claim~\ref{claim:kt distance}}
\label{appendix:proof of kt distance}
Without loss of generality, we assume the linear order $R$ is \begin{align*}
    R = a_1 \succ a_2 \succ \cdots a_{m-1} \succ a_m.
\end{align*}
For any $1 \le i < j \le m$, we have $a_i \succ a_j$ in $R$. Now we calculate the pairwise disagreement between $R$ and the parameter profile $P^\Theta_G$. Note that $\wmg(P^\Theta_G) = G = (V,E)$ where $G$ is an unweighted directed graph. Fix $1 \le i < j \le m$, the pairwise comparison between alternatives $a_i$ and $a_j$ lies in one of the three cases below.
\paragraph{Case 1: There is no edge between $a_i$ and $a_j$ in $G$.} In this case we know the weight on $a_i \succ a_j$ and the weight on $a_j \succ a_i$ are equal to $\frac{|P^\Theta_G|}{2}$.
\paragraph{Case 2:  There is an edge $a_i \rightarrow a_j$ in $G$.} In this case we know the weight on $a_j \succ a_i$ is $\frac{|P^\Theta_G| - 1}{2}$.

\paragraph{Case 3:  There is an edge $a_j \rightarrow a_i$ in $G$.} In this case we know the weight on $a_j \succ a_i$ is $\frac{|P^\Theta_G| + 1}{2}$.

Since there are $|E|$ arcs in $G$ and $f(G,R)$ backward arcs of $R$ in $G$, we know there are $\binom{m}{2} - |E|$ pairwise comparisons in Case 1, $|E| - f(G,R)$ pairwise comparisons in Case 2, and $f(G,R)$ pairwise comparisons in Case 3. Summing the weights in all cases, we get 
\begin{align*}
    \kt\left(P_G^\Theta,R\right) &= \frac{|P^\Theta_G|}{2} \left( \binom{m}{2} - |E| \right) + \frac{|P^\Theta_G|-1}{2} (|E| - f(G,R)) + \frac{|P^\Theta_G|+1}{2} f(G,R) \\
    &= \frac{|P^\Theta_G|}{2} \binom{m}{2} - \frac{|E|}{2} + f(G,R)
\end{align*}
Therefore, we conclude by definition of $M$ that
\begin{align*}
    \kt\left(P_G^\Theta \cdot \frac{m^K}{|P^\Theta_G|},R\right) = \frac{m^K}{|P^\Theta_G|} \cdot \kt\left(P_G^\Theta,R\right) = M + \frac{m^K}{|P^\Theta_G|} \cdot  f(G,R).
\end{align*}